\documentclass[aps,pra,twocolumn]{revtex4-2}
\usepackage{amsmath,amsfonts,amssymb,amsthm}
\usepackage{mathtools}
\usepackage[pdftex]{graphicx}
\usepackage[usenames, dvipsnames, svgnames, table]{xcolor}
\usepackage[colorlinks=true, citecolor=Blue, linkcolor=BrickRed, urlcolor=Brown]{hyperref}
\usepackage{txfonts}
\usepackage{mathrsfs}
\usepackage{bm}
\usepackage{multirow}
\usepackage{braket}
\usepackage{import}
\usepackage{qcircuit}
\usepackage[ruled, noline, noend]{algorithm2e}
\usepackage{array}
\usepackage{comment}
\usepackage{tikz}
\usepackage{tikz-cd}
\tikzcdset{scale cd/.style={every label/.append style={scale=0.8 * #1},
    cells={nodes={scale=0.9 * #1}}}}

\newtheorem{theorem}{Theorem}
\newtheorem{lemma}{Lemma}

\theoremstyle{definition}
\newtheorem{definition}{Def.}
\newtheorem*{definitionRestate}{Def}
\newcommand{\be}{\begin{equation}}
\newcommand{\ee}{\end{equation}}
\newcommand{\ben}{\begin{eqnarray}}
\newcommand{\een}{\end{eqnarray}}
\newcommand{\bes}{\begin{subequations}}
\newcommand{\ees}{\end{subequations}}
\newcommand{\bF}{\begin{figure}}
\newcommand{\eF}{\end{figure}}

\DeclareMathAlphabet{\pazocal}{OMS}{zplm}{m}{n}

\newcommand{\orcid}[1]{\href{https://orcid.org/#1}{\includegraphics[height = 2ex]{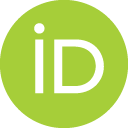}}}

\begin{document}

\title{Phase Transitions in Decision Problems Over Odd-Sized Alphabets}

\author{Andrew Jackson \orcid{0000-0002-5981-1604}}
\affiliation{School of Informatics, University of Edinburgh, Edinburgh, EH8 9AB, United Kingdom}
\date{\today}


\begin{abstract}
     In [A. Jackson, Explaining the ubiquity of phase transitions in decision problems (2025), arXiv:2501.14569], I established that phase transitions are always present in a large subset of decision problems over even-sized alphabets, explaining -- in part -- why phase transitions are seen so often in decision problems. However, decision problems over odd-sized alphabets were not discussed. Here, I correct that oversight, showing that a similar subset of decision problems over odd-sized alphabets also always exhibit phase transitions.
\end{abstract}

\maketitle
\section{Introduction}
\label{IntroSec}
\subsection{Background and Decision Problems}
Decision problems (as defined in Def.~\ref{def:decisionProb}) and the analysis of their behaviors are the cornerstones of theoretical computer science, and form the basis for our most foundational notions of computing, such as: Turing machines~\cite{Arora_Barak_2009_11-23} and computational complexity theory~\cite{sep-computational-complexity}.

In fact, for most, if not all, observed phenomena in computing, the most central, simple, and important aspect of it can be found through an examination of decision problems. The insights gained from the core of the phenomena, found in   decision problems, often cascade back up to the more complex instances of that phenomena, found in more general situations. In this way, barriers to investigating important but inscrutable phenomena can potentially be circumvented.  

This has been a fruitful approach to investigating important questions in computer science e.g. the question of whether quantum computers~\cite{NQCC_2025} provide any advantage over their classical counterparts is often formalized in terms of comparing BQP~\cite{10.1137/S0097539796300921, aaronson2009bqp} (the efficiently solvable -- defined as in Ref.~\cite{edmonds1965} and Ref.~\cite{Cobham1965} -- decision problems on a quantum computer) and P~\cite{sipser13P} (the efficiently solvable decision problems on a classical computer). Similarly, a vast array of important questions in computer science have been reduced to comparing NP~\cite{sipser13NP} (the efficiently solvable decision problems on a non-deterministic~\cite{sipser13NonDet} classical computer) and P.

In this paper, decision problems will, without loss of generality, be assumed to consist of deciding if a given word is in a specific language. The preceding sentence is given formal meaning using the below Def.~\ref{def:decisionProb}, Def.~\ref{def:alphabet}, Def.~\ref{def:SigmaStar}, and Def.~\ref{def:wordAndLanguage}.
\begin{definition}
    \label{def:decisionProb}
    A \underline{decision problem} is any problem where there are only two possible answers. Typically, ACCEPT and REJECT.
\end{definition}
\begin{definition}
    \label{def:alphabet}
    An \underline{alphabet} is a finite set of symbols e.g. $\{a, b, c, d, e, f, g\}$. Herein, I will assume all alphabets considered have at least two elements.
\end{definition}
\begin{definition}
    \label{def:SigmaStar}
    For any alphabet, $\Sigma$, define \underline{$\Sigma^*$} as the set of all finite strings of symbols from $\Sigma$.
    \end{definition}
    \begin{definition}
    \label{def:wordAndLanguage}
    A \underline{word} over the alphabet $\Sigma$ is an element of $\Sigma^*$ and any subset of $\Sigma^*$ is referred to as a \underline{language over $\Sigma$}.
\end{definition}

One phenomena that appears exceedingly often in natural decision problems is phase transitions~\cite{minesweeperPhase, https://doi.org/10.1002/(SICI)1098-2418(1999010)14:1<63::AID-RSA3>3.0.CO;2-7, Gent1994TheSP, stein_newman_2013, doi:10.1080/0022250X.1982.9989929}.
Most generally, phase transitions in decision problems are defined informally as a rapid change in the probability of being in a specific language as a specific polynomial-time function -- defined on the Kleene star operation of the relevant alphabet -- changes. This rapid change in one aspect of the problem across a relatively small change in another aspect is reminiscent of phase transitions in physical systems and so it is hoped that the most interesting and critical aspects of phase transitions in many-body / condensed-matter systems are captured by phase transitions in decision problems.

An example of how these phase transitions appear is given in Fig.~\ref{fig:ExampleImage}.
    \begin{figure}[h!]
    \centering
\includegraphics[width=0.49\textwidth] {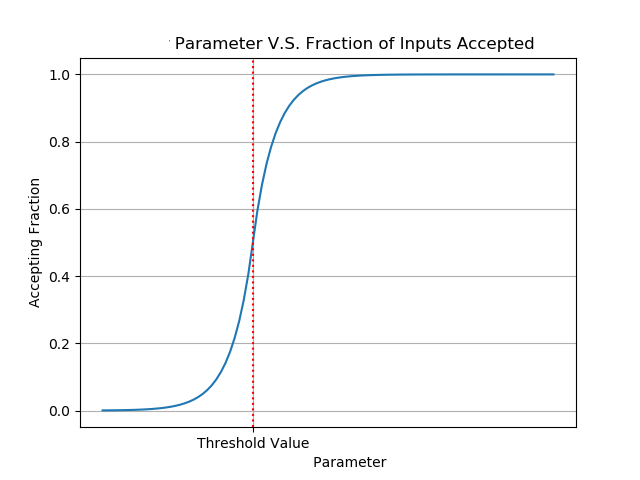} 
    \caption{[From Ref.~\cite{jackson2025explainingubiquityphasetransitions}] A typical example of a phase transition. The defining feature of a phase transition -- in decision problems -- is the change in the accepting fraction as the parameter -- a real-valued polynomial time computable function of an instance of the problem -- approaches the threshold value (the vertical red line).}
    \label{fig:ExampleImage}
\end{figure}
Figures similar to Fig.~\ref{fig:ExampleImage} can be found in Ref.~\cite{minesweeperPhase},  Ref.~\cite{Gent1994TheSP}, Ref.~\cite{101371journalpone0215309}, Ref.~\cite{cross3Sat}, and Ref.~\cite{BAILEY20071627}; further showing how common phase transitions are.

However, for mathematical investigations of phase transitions, mere examples of phase transitions will not suffice, and a formal definition is required. In Ref.~\cite{jackson2025explainingubiquityphasetransitions}, I provided one such definition of phase transitions. But before I can recite it, I must specify the concepts phase transitions are defined in terms of. 

\subsection{Phase Transition Preliminary Definitions}
\label{Def.sPrelimSubsection}
The most basic and definitional aspect of a phase transition is a change in the properties of a system. In decision problems, there are considerably fewer obviously meaningful properties that may significantly change than in many-body systems (e.g. in condensed matter systems, phase transitions can manifest as rapid changes in magnetic ordering~\cite{hook1991solidMagnetic, Buschow2003}, crystal structure~\cite{hook1991solid}, or electrical conductivity~\cite{hook1991solidSuper}), where phase transitions are more famous and widely studied. For decision problems, the quantity that changes in the course of a phase transition is the accepting fraction, defined in Def.~\ref{def:acceptingFraction}.
\begin{definition}
\label{def:acceptingFraction}
    For any alphabet, $\Sigma$, and subset, $\mathcal{S} \subseteq \Sigma^*$, the \underline{accepting fraction}, $\mathcal{A}[\mathcal{S}]$, relative to a language $\mathcal{L} \subseteq \Sigma^*$, is the fraction of $\mathcal{S}$ that is in $\mathcal{L}$.
    Mathematically, $\mathcal{A}[\mathcal{S}]$ may be expressed as: $\forall \mathcal{S} \subseteq \Sigma^*$,
    \begin{align}
        \label{eqn:indef:acceptingFraction}
        \mathcal{A}[\mathcal{S}] 
        &=
        \dfrac{\big \vert \mathcal{S} \cap \mathcal{L} \big \vert}{\big \vert \mathcal{S} \big \vert}.
    \end{align}
\end{definition}
With the quantity that changes dependently as a phase transition occurs defined, in Def.~\ref{def:acceptingFraction}, I then turn to ask about the changing quantity that induces the phase transition. In many-body or condensed-matter physics, this is most stereotypically the temperature (e.g. a rise in temperature melting ice). But decision problems do not have a temperature. Instead the changing -- and change inducing -- quantity in decision problem phase transitions is a parameter (a polynomial-time computable function, as defined in Def.~\ref{orderParameterDef}) that is specific to the language that the phase transition is observed in.
\begin{definition}
    \label{orderParameterDef}
    A \underline{parameter} is any polynomial time computable mapping from $\Sigma^*$ to $\mathbb{R}$, where $\Sigma$ may be any alphabet. 
\end{definition}
To be clear, the quantity that induces the phase transition is the value -- from $\mathbb{R}$ -- of the parameter; it plays the role corresponding to the temperature in the above mentioned example of melting ice.

In order for the accepting fraction -- that changes during the phase transition -- and the parameter -- that causes the phase transition by changing -- to be more easily related, I specify some notation, in Def.~\ref{def:parameterSlice}, that will be used throughout this paper.
\begin{definition}
    \label{def:parameterSlice}
    For any parameter, $\gamma: \Sigma^* \longrightarrow \mathbb{R}$, the \underline{parameter slice}, $\mathcal{S}^{\gamma}_{n} \subseteq \Sigma^*$, is defined by: $\forall n \in \mathbb{R}$,
    \begin{align}
        \mathcal{S}^{\gamma}_n 
        =
        \big\{ x \in \Sigma^* \text{ } \vert \text{ } \gamma (x) = n \big\}.
    \end{align}
    I.e. $\mathcal{S}^{\gamma}_{n}$ is the set of all $x \in \Sigma^*$ such that the parameter, $\gamma$, takes the value $n \in \mathbb{R}$.
\end{definition}
\subsection{Decision Problem Phase Transition Formal Definition}
I can now formally define phase transitions (in Def.~\ref{PhaseDef}), aiming to capture a slightly more broad notion of phase transitions than the one observed in Fig.~\ref{fig:ExampleImage}. 
\begin{definition}[From Ref.~\cite{jackson2025explainingubiquityphasetransitions}]
\label{PhaseDef}
    A language, $\mathcal{L} \subseteq \Sigma^*$, exhibits a \underline{phase transition} if and only if there exists a parameter, $\gamma: \Sigma^* \longrightarrow \mathbb{R}$, such that:
    \begin{enumerate}
        \item As $n \longrightarrow \infty$, $\mathcal{A}[\mathcal{S}^{\gamma}_n]\longrightarrow 1$, with a monotonically increasing lower bound.
    \item As $n \longrightarrow -\infty$, $\mathcal{A}[\mathcal{S}^{\gamma}_n]\longrightarrow 0$, with a monotonically decreasing upper bound.   
    \item  The fraction of $\Sigma^*$ that takes a value of $\gamma$ between A $\in \mathbb{R}^+$ and A + $\delta$ (where $\delta \in \mathbb{R}^+$) grows exponentially~\footnote{or at least does beyond a short distance from the threshold value} as $\vert$ A - $\mathcal{T}$ $\vert$ increases (where $\mathcal{T} \in \mathbb{R}$ is a specific value).
    \end{enumerate}
\end{definition}
It may be useful to note that the above definition of phase transitions does \emph{not} require the change in acceptance fraction happen while the parameter transitions through an interval (of the real numbers) of a specific size; this is as, if such a requirement were added, any language meeting the above Def.~\ref{PhaseDef} can have its parameter scaled by some constant real value to meet the new definition. So adding a requirement on how quickly the transition must happen is redundant.  
\subsection{P-Isomorphism Essentials}
\label{sec:PisomorpStuff}
Before going any further, I pause to present a core component of most arguments in this paper: P-isomorphisms. In Sec.~\ref{sec:PisomorpStuff}, I also present important related concepts. The first definition is Def.~\ref{def:PIsomorph}.
\begin{definition}
    \label{def:PIsomorph}
    For any alphabets, $\Sigma$ and $\Pi$, a \underline{P-isomorphism} is any polynomial-time, bijective mapping from $\Sigma^*$ to $\Pi^*$ that can also be inverted (i.e. the input can be retrieved from the corresponding output) in polynomial time.
\end{definition}
With P-isomorphisms defined, I can then define another important concept, in Def.~\ref{pIsoDef}.
\begin{definition}
    \label{pIsoDef}
    For any pair of alphabets, $\Sigma$ and $\Pi$, and any P-isomorphism, $\xi: \Sigma^* \longrightarrow \Pi^*$, the \underline{P-isomorphism output size} of $\xi$, denoted $N_{\xi}: \Sigma^* \longrightarrow \mathbb{N}_0$, is defined by: $\forall x \in \Sigma^*$,
    \begin{align}
        N_{\xi}(x)
        &=
        \big \vert \xi(x) \big \vert,
    \end{align}
    where $\big \vert \cdot \big \vert : \Pi^* \longrightarrow \mathbb{N}_0$ is the length of its argument word~\footnote{In fact, throughout this paper I will regularly use this function and define it polymorphically to return the length of a word regardless of the alphabet the word is formed from.}.
\end{definition}
A slight enhancement of P-isomorphisms are preserving-P-isomorphisms, defined in Def.~\ref{def:preservingPIso}.
\begin{definition}
    \label{def:preservingPIso}
    For any alphabets, $\Sigma$ and $\Pi$, a \underline{preserving-P-isomorphism}, $\xi: \Sigma^* \longrightarrow \Pi^*$, from $\mathcal{L} \subseteq \Sigma^*$ to $\mathcal{H} \subseteq \Pi^*$ is a P-isomorphism from $\Sigma^*$ to $\Pi^*$ such that: $\forall x \in \Sigma^*$,
    \begin{align}
        \label{eqn:PreservingAspect}
        x \in \mathcal{L} \iff \xi(x) \in \mathcal{H}.
    \end{align}
\end{definition}
The situation that Eqn.~\ref{eqn:PreservingAspect} specifies is depicted in Fig.~\ref{fig:Lemma3Diagram}.

The final definition of Sec.~\ref{sec:PisomorpStuff} is a further development of the chain of developments above. Defining a set based on having the same P-isomorphism output size: Def.~\ref{def:mathcalBnphi}. 
\begin{definition}
    \label{def:mathcalBnphi}
    For any alphabets, $\Sigma$ and $\Pi$, and P-isomorphism, $\phi: \Sigma^* \longrightarrow \Pi^*$, the set
    \underline{$\mathcal{B}_n^{\phi}$} $\subseteq \Sigma^*$ is defined by:
   \begin{align} 
        \mathcal{B}_n^{\phi}~=~\big \{ x \in \Sigma^* \text{ } \vert \text{ } N_{\phi}(x) = n \big\},
    \end{align}
    where $N_{\phi}: \Sigma^* \longrightarrow \mathbb{N}$ is the P-isomorphism output length of $\phi$, as defined in Def.~\ref{pIsoDef}.
\end{definition}

Def.~\ref{def:parameterSlice} and Def.~\ref{def:mathcalBnphi} may be combined to express $\mathcal{B}_n^{\phi}$ as: for any pair of alphabets, $\Sigma$ and $\Pi$, and  P-isomorphism, $\phi: \Sigma^* \longrightarrow \Pi^*$,
\begin{align} 
        \mathcal{B}_n^{\phi}
        =
        \big \{ x \in \Sigma^* \text{ } \vert \text{ } N_{\phi}(x) = n \big\}
        =
        \mathcal{S}^{N_{\phi}}_n
        \subseteq 
        \Sigma^*.
    \end{align}

\subsection{Pre-existing Results on Paddability, Not-Anywhere-Exponentially-unbalanced languages, and Phase Transitions}
Like in Ref.~\cite{jackson2025explainingubiquityphasetransitions}, a key pre-existing result is the definition of a complexity class known as RoughP and the fact that all paddable languages (defined in Def.~\ref{padDef}) are in it. RoughP is defined as in Def.~\ref{def:RoughP}.
\begin{definition}{(Def.~3 in Ref.~\cite{farago2016roughly})}
    \label{def:RoughP}
    Let,
    \begin{enumerate}
        \item $\Sigma$ be an alphabet with $\vert \Sigma \vert \geq 2$
    \item $\mathcal{L} \subseteq \Sigma^*$
be a language.
    \end{enumerate}
Then, $\mathcal{L} \in$ \underline{RoughP}, if and only if there exists a P-isomorphism, $\phi: \Sigma^* \longrightarrow \Sigma^*$, and a polynomial time algorithm,
$\mathcal{P}$ : $\Sigma^* \longrightarrow$ \{Accept, Reject, $\perp$\}, such that:
\begin{enumerate}
    \item $\mathcal{P}$ correctly decides $\mathcal{L}$, as an errorless heuristic. That is, it never outputs a wrong decision:
if $\mathcal{P}$ accepts a string $x \in \Sigma^*$, then $x$  $\in \mathcal{L}$ always holds, and if $\mathcal{P}$ rejects $x \in \Sigma^*$, then $x \not \in \mathcal{L}$ always holds.
\item Besides Accept/Reject, $\mathcal{P}$ may output another symbol, $\perp$, meaning it is unable to decide if the input is in the language.
This can occur, however, only for at most an exponentially small fraction of strings. I.e. there is a constant $c \in [0,1)$ such that: $\forall n \in \mathbb{N}$,
\begin{align}
    \label{eqn:RoughPDefiningEquation}
    \dfrac{\vert \mathcal{B}_n^{\phi} \cap \{ x \in \Sigma^* \text{ }\vert\text{ } \mathcal{P}(x) = \perp\} \vert}{\vert \mathcal{B}_n^{\phi} \vert} \leq c^n,
\end{align}
where \underline{$\mathcal{B}_n^{\phi}$} is as in Def.~\ref{def:mathcalBnphi}.
\end{enumerate}
\end{definition}

\emph{The} key relevant result from previous work of interest for my purposes herein is by the current author and found in Ref.~\cite{jackson2025explainingubiquityphasetransitions}: Theorem~\ref{mainResultTheoremOFLASTPAPER}.
\begin{theorem}[Theorem 2 in Ref.~\cite{jackson2025explainingubiquityphasetransitions}]
\label{mainResultTheoremOFLASTPAPER}
    Any paddable not-anywhere-exponentially-unbalanced language over an even-sized alphabet exhibits a phase transition.
\end{theorem}
There a few terms in Theorem~\ref{mainResultTheoremOFLASTPAPER} that are not yet defined. I therefore provide the below definitions of them.
\begin{definition}
    \label{padDef}
    A language, $\mathcal{L} \subseteq \Sigma^*$, is \underline{paddable} if and only if there exists two polynomial time computable functions:
\begin{enumerate}
        \item $\textit{Pad}: \Sigma^* \times \Sigma^* \longrightarrow \Sigma^*$,
        \item $\textit{Dec}:  \Sigma^* \longrightarrow \Sigma^*$,
\end{enumerate}
    such that, $\forall x, y \in \Sigma^*$:
    \begin{enumerate}
        \item $\textit{Pad}(x, y) \in \mathcal{L} \iff x \in \mathcal{L}$,
        \item $\textit{Dec}(\textit{Pad}(x, y)) = y$.
    \end{enumerate}
\end{definition}

\begin{definition}
    \label{PhiBalanced}
    Not-anywhere-exponentially-unbalanced languages are a subset of paddable languages.
    
    A language is \underline{not-anywhere-exponentially-unbalanced} if there exists some polynomial, $\textit{Poly}: \mathbb{N} \longrightarrow \mathbb{R}$, such that $\forall n \in \mathbb{N}$, neither the fraction of $\mathcal{B}_n^{\phi}$ that is in the language (and $\mathcal{P}$ decides correctly) nor the fraction not in the language (and $\mathcal{P}$ decides correctly) are less than $\big( \textit{Poly}(n) \big)^{-1}$, and $\textit{Poly}(n) \cdot \big( 1 / \sqrt{2} \big)^n$ is monotonically decreasing.

    If the above holds for a paddable language when ``$\mathcal{P}$" is replaced by any RoughP algorithm constructed by applying a preserving-P-isomorphism  -- to a paddable language over an alphabet larger, by one, than the current alphabet -- to the input string so that a RoughP algorithm for the new language -- constructed as in Ref.~\cite{farago2016roughly} -- can be applied, I call the language \underline{alt-NAEU}.
\end{definition} 
\section{Aims and Main Results}
\subsection{Adequately-Balanced Languages}
Unfortunately, not-anywhere-exponentially-unbalanced languages cannot (at least by me) be shown to \emph{always} exhibit a phase transition. In Ref.~\cite{jackson2025explainingubiquityphasetransitions}, it was shown to suffice for even-sized languages but considering odd-sized languages prevents it from being sufficient. For a phase transition to be ensured, the languages is required to be adequately-balanced, as defined in Def.~\ref{def:adequete}.
\begin{definition}
    \label{def:adequete}
    A language, over alphabet $\Sigma$, is \underline{adequately-balanced} if it is paddable and neither the fraction of $\mathcal{B}_n^{\phi}$ that is in the language (and $\mathcal{P}$ decides correctly) nor the fraction not in the language (and $\mathcal{P}$ decides correctly) are less than some polynomial, $\big( \textit{Poly}(n) \big)^{-1}$, and $\textit{Poly}(n) \cdot \big( 1 / \sqrt{2} \big)^n$ is monotonically decreasing.

    For languages where $\vert \Sigma \vert$ is odd, there are additional requirements. These are:
    \begin{enumerate}
        \item The distribution of elements of $\mathcal{B}^{\phi}_n \subset \Sigma^*$ that are both in and not in the language are split proportionately between the subsets in differing $\xi^{-1} \big( \mathcal{B}^{\xi^{-1} \circ \phi \circ \xi}_n \big)$ -- so as to maintain the not-anywhere-exponentially-unbalanced property of each continuous subset -- (with different values of $n \in \mathbb{N}_0$), where $\xi$ is as in Def.~\ref{def:DefingXi}.
        \item $\textit{Poly}$ is monotonically increasing and, $\forall n \in \mathbb{N}_0$,
        \begin{align} 
            \label{eqn:PolyRequirement}
            \dfrac{\partial}{ \partial x} \bigg( \textit{Poly}(x) \bigg) \bigg \vert_{x = \lambda n}
        &\leq
         \dfrac{\big \vert \ln{\big( 1 / \sqrt{2} \big)} \big \vert}{\lambda} \textit{Poly} \big( \lambda n \big),
        \end{align}
        where $\lambda = 2 \log_{\vert \Sigma \vert} \big( \vert \Sigma \vert + 1 \big) $.
        \item The language is alt-NAEU.
    \end{enumerate}
\end{definition}

\noindent \underline{\textbf{Example Functions Meeting the Requirements of \textit{Poly}}}\\
The condition in Eqn.~\ref{eqn:PolyRequirement} for the required polynomial, $\textit{Poly}$, is met if:
\begin{align}
    \textit{Poly}(n)
    &=
    \beta n + \gamma,
\end{align}
where $\beta, \gamma \in \mathbb{R}^+$ such that $\gamma \geq 4 \beta$.
\subsection{Statement of Results}
The main focus of the present paper is to prove Theorem~\ref{lem:MainLemma}.
\begin{theorem}
    \label{lem:MainLemma}
    Any adequately-balanced language over an odd-sized alphabet exhibits a phase transition.
\end{theorem}
Theorem~\ref{lem:MainLemma} is proven in Sec.~\ref{sec:proofOfMainLemma} and is useful as, in combination with Theorem~\ref{mainResultTheoremOFLASTPAPER}, it implies Theorem~\ref{mainResultTheorem}.
\begin{theorem}
\label{mainResultTheorem}
    Any adequately-balanced language exhibits a phase transition.
\end{theorem}
\begin{proof}
    For any adequately-balanced language over an alphabet, $\Sigma$; either $\vert \Sigma \vert$ is even, in which case Theorem~\ref{mainResultTheoremOFLASTPAPER} implies it exhibits a phase transition, or $\vert \Sigma \vert$ is odd, in which case Theorem~\ref{lem:MainLemma} implies it exhibits a phase transition. 
\end{proof}
Theorem~\ref{mainResultTheorem} furthers the goal of Ref.~\cite{jackson2025explainingubiquityphasetransitions}, helping to explain the ubiquity of phase transitions in decision problems even further.

Due to the previously-mentioned habit of phenomena in decision problems reflecting their more complicated and hard-to-study analogues in computing -- and perhaps physics -- more generally, there is good reason to believe this result may aid in identifying, explaining, and classifying phase transitions in more general systems, where understanding emergent behaviours such as phase transitions is more immediately and obviously important (e.g. it is clearly important to understand the melting points of various metals when designing safety-critical systems that will operate in high-temperature environments).   

\subsection{Main Result: Proof of Theorem~\ref{lem:MainLemma}}
\label{sec:proofOfMainLemma}
For readability and clarity, instead of disrupting the flow of the paper with a very long proof, the proof of Theorem~\ref{lem:MainLemma} presented below instead relies on Lemma~\ref{lem:PIsoExists} and Lemma~\ref{lem:PisoImpliesPhase}, both of which are proved in Appendix~\ref{app:LemsForMainLemma}. However, the proof of Theorem~\ref{lem:MainLemma} depends most crucially on Theorem~\ref{mainResultTheoremOFLASTPAPER}, and uses the phase transitions in paddable not-anywhere-exponentially-unbalanced language over an even-sized alphabets to construct phase transitions in the equivalent languages over odd-sized alphabets via preserving-P-isomorphisms. 

\begin{proof}[Proof of Theorem~\ref{lem:MainLemma}]
    Let $\mathcal{L}$ be an adequately-balanced language over an \emph{odd}-sized alphabet.
    Using Lemma~\ref{lem:PIsoExists} (in Appendix~\ref{app:LemsForMainLemma}), $\mathcal{L}$ must be preserving-P-isomorphic to a paddable not-anywhere-exponentially-unbalanced language over an \emph{even}-sized alphabet, which I refer to as $\mathcal{H}$.
    
    As $\mathcal{H}$ is a paddable not-anywhere-exponentially-unbalanced language over an even-sized alphabet, Theorem~\ref{mainResultTheoremOFLASTPAPER} implies that $\mathcal{H}$ exhibits a phase transition.
    Therefore, due to Lemma~\ref{lem:PisoImpliesPhase} (in Appendix~\ref{app:LemsForMainLemma}), $\mathcal{L}$ must exhibit a phase transition as it is preserving-P-isomorphic to a language, $\mathcal{H}$, that exhibits a phase transition.
\end{proof}

\section{Discussion}
Herein I have furthered the work of Ref.~\cite{jackson2025explainingubiquityphasetransitions}: showing that all adequately-balanced (as defined in Def.~\ref{def:adequete}) languages exhibit a phase transition. Before, due to Ref.~\cite{jackson2025explainingubiquityphasetransitions}, this was only known to hold for languages over even-sized languages. Therefore, the demonstration, in Theorem~\ref{lem:MainLemma}, that all adequately-balanced languages over odd-sized alphabets exhibit a phase transition entails that all adequately-balanced languages exhibit a phase transition (as in Theorem~\ref{mainResultTheorem}).

Given this paper served to relax the requirement -- present in Ref.~\cite{jackson2025explainingubiquityphasetransitions} -- that the languages considered must be over even-sized alphabets, it is natural to ask how much further the restrictions used in this paper can be relaxed. 
The prime candidate for elimination or relaxation, in Theorem~\ref{mainResultTheorem}, is the assumption that the languages shown to exhibit phase transitions are required to be not-anywhere-exponentially-unbalanced or the conditions on the associated polynomial function. This is, in part, as sparsity~\cite{sparseCitation} (defined in Def.~\ref{def:sparseLanguages} and admittedly a slightly different assumption to being not-anywhere-exponentially-unbalanced) of a language is incompatible with paddability (the other required condition of the languages shown to exhibit phase transitions in Theorem~\ref{mainResultTheorem}). This is demonstrated in Appendix~\ref{app:PaddingSparsity}. However, I leave the possibility of relaxing this condition -- perhaps by showing that paddability imposes even stricter limits on how dense a language must be -- open for future work.  

The other assumptions required -- only of languages over odd alphabets -- in the definition of a language being adequately balanced are mostly present to enable the proof techniques used in this paper. There is no reason to believe they are fundamental restrictions on when phase transitions can occur, and so it is likely that these conditions can be significantly relaxed but that may require a more advanced array of techniques than the ones used herein.

As there are promise-BQP~\footnote{the complexity class corresponding to practical quantum computing~\cite{Arora_Barak_2009_201-236}} languages (and promise-BQP-complete languages) known to be paddable~\cite{jackson2024extensivelypbiimmunepromisebqpcompletelanguages}, it is feasible that there may be applications of this work to the verification of quantum computations (both digital~\cite{Barz2013, v012a003, Hangleiter_2017, Kashefi_2017,  Gheorghiu_2018, Ferracin_2019, Markham_2020, Ferracin_2021, jackson2025accreditationlimitedadversarialnoise} and analogue~\cite{Shaffer2021, doi:10.1073/pnas.2309627121, jackson2025improvedaccreditationanaloguequantum}). By building decision problems based on the output of those computations and, assuming they can be contrived to be paddable and not-anywhere-exponentially-unbalanced, using their phase transitions as a heuristic to check the outputs of the computations.

A protocol for such a task may appear as in Protocol~\ref{PhaseVerificationProtocol}.

\begin{figure}
    \centering
\begin{algorithm}[H]

\SetAlgorithmName{Protocol}{protocol}{List of Protocols}
\noindent \underline{\textbf{INPUTS:}}\\
$\textbf{P}$: A promise decision problem decidable via a BQP-device\\
$\textit{p}$: A instance of a valid input to the problem $\textbf{P}$\\
\hrulefill

\begin{enumerate}
    \item Construct a set, $\mathcal{P}_{\textit{p}}$, of instances of a valid input to the problem $\textbf{P}$ expected to experience comparable error --  during the execution of the circuit to decide it on the BQP-device -- to $\textit{p}$
    \item Use a phase transition of $\textbf{P}$ to obtain solutions, $\{ S'_j \}^{\vert \mathcal{P}_{\textit{p}}\vert }_{j = 1}$ predicted by it, and corresponding confidences, $\{ C_j \}^{\vert \mathcal{P}_{\textit{p}}\vert }_{j = 1}$, for each instance in $\mathcal{P}_{\textit{p}}$
    \item Execute the -- potentially erroneous -- circuits, on the BQP-device, to decide each instance in $\mathcal{P}_{\textit{p}}$. Call the outputs $\{ S_j \}^{\vert \mathcal{P}_{\textit{p}}\vert }_{j = 1}$ and assume $S_0$ is the BQP-device's solution to $\textit{p}$
    \item Initialize a variable, \textit{overall\_confidence} = 1
    \item \textbf{For} $S_j$ \textbf{in} $\{ S_j \}^{\vert \mathcal{P}_{\textit{p}}\vert }_{j = 1}$:
    \begin{enumerate}
        \item \textbf{If} ( $S_j \not = S_j'$ ):
        \begin{enumerate}
            \item \textit{overall\_confidence} *= $(1-C_j)$
        \end{enumerate}
    \end{enumerate}
\end{enumerate}
\noindent \hrulefill\\
    \noindent \underline{\textbf{RETURN:}} \\
    $S_0$ (the result the BQP-device gives for $\textit{p}$) and \textit{overall\_confidence}
\caption{Suggested protocol sketch for verifying promise-BQP languages.
 \label{PhaseVerificationProtocol}}
\end{algorithm}
\end{figure}
I leave the further and proper development of this approach to future work.

\section{Acknowledgements}
The author acknowledges the support of the Quantum Advantage Pathfinder (EP/X026167/1).
\newpage

\bibliography{References}

\begin{thebibliography}{53}%
\makeatletter
\providecommand \@ifxundefined [1]{%
 \@ifx{#1\undefined}
}%
\providecommand \@ifnum [1]{%
 \ifnum #1\expandafter \@firstoftwo
 \else \expandafter \@secondoftwo
 \fi
}%
\providecommand \@ifx [1]{%
 \ifx #1\expandafter \@firstoftwo
 \else \expandafter \@secondoftwo
 \fi
}%
\providecommand \natexlab [1]{#1}%
\providecommand \enquote  [1]{``#1''}%
\providecommand \bibnamefont  [1]{#1}%
\providecommand \bibfnamefont [1]{#1}%
\providecommand \citenamefont [1]{#1}%
\providecommand \href@noop [0]{\@secondoftwo}%
\providecommand \href [0]{\begingroup \@sanitize@url \@href}%
\providecommand \@href[1]{\@@startlink{#1}\@@href}%
\providecommand \@@href[1]{\endgroup#1\@@endlink}%
\providecommand \@sanitize@url [0]{\catcode `\\12\catcode `\$12\catcode `\&12\catcode `\#12\catcode `\^12\catcode `\_12\catcode `\%12\relax}%
\providecommand \@@startlink[1]{}%
\providecommand \@@endlink[0]{}%
\providecommand \url  [0]{\begingroup\@sanitize@url \@url }%
\providecommand \@url [1]{\endgroup\@href {#1}{\urlprefix }}%
\providecommand \urlprefix  [0]{URL }%
\providecommand \Eprint [0]{\href }%
\providecommand \doibase [0]{https://doi.org/}%
\providecommand \selectlanguage [0]{\@gobble}%
\providecommand \bibinfo  [0]{\@secondoftwo}%
\providecommand \bibfield  [0]{\@secondoftwo}%
\providecommand \translation [1]{[#1]}%
\providecommand \BibitemOpen [0]{}%
\providecommand \bibitemStop [0]{}%
\providecommand \bibitemNoStop [0]{.\EOS\space}%
\providecommand \EOS [0]{\spacefactor3000\relax}%
\providecommand \BibitemShut  [1]{\csname bibitem#1\endcsname}%
\let\auto@bib@innerbib\@empty
\bibitem [{\citenamefont {Arora}\ and\ \citenamefont {Barak}(2009{\natexlab{a}})}]{Arora_Barak_2009_11-23}%
  \BibitemOpen
  \bibfield  {author} {\bibinfo {author} {\bibfnamefont {S.}~\bibnamefont {Arora}}\ and\ \bibinfo {author} {\bibfnamefont {B.}~\bibnamefont {Barak}},\ }\href@noop {} {\emph {\bibinfo {title} {Computational Complexity: A Modern Approach}}}\ (\bibinfo  {publisher} {Cambridge University Press},\ \bibinfo {year} {2009})\ pp.\ \bibinfo {pages} {11--23}\BibitemShut {NoStop}%
\bibitem [{\citenamefont {Dean}(2021)}]{sep-computational-complexity}%
  \BibitemOpen
  \bibfield  {author} {\bibinfo {author} {\bibfnamefont {W.}~\bibnamefont {Dean}},\ }\bibfield  {title} {\bibinfo {title} {{Computational Complexity Theory}},\ }in\ \href@noop {} {\emph {\bibinfo {booktitle} {The {Stanford} Encyclopedia of Philosophy}}},\ \bibinfo {editor} {edited by\ \bibinfo {editor} {\bibfnamefont {E.~N.}\ \bibnamefont {Zalta}}}\ (\bibinfo  {publisher} {Metaphysics Research Lab, Stanford University},\ \bibinfo {year} {2021})\ \bibinfo {edition} {{F}all 2021}\ ed.\BibitemShut {Stop}%
\bibitem [{\citenamefont {{The National Quantum Computing Centre}}(2025)}]{NQCC_2025}%
  \BibitemOpen
  \bibfield  {author} {\bibinfo {author} {\bibnamefont {{The National Quantum Computing Centre}}},\ }\href {https://www.nqcc.ac.uk/resources/what-is-quantum-computing/} {\bibinfo {title} {What is quantum computing?}} (\bibinfo {year} {2025})\BibitemShut {NoStop}%
\bibitem [{\citenamefont {Bernstein}\ and\ \citenamefont {Vazirani}(1997)}]{10.1137/S0097539796300921}%
  \BibitemOpen
  \bibfield  {author} {\bibinfo {author} {\bibfnamefont {E.}~\bibnamefont {Bernstein}}\ and\ \bibinfo {author} {\bibfnamefont {U.}~\bibnamefont {Vazirani}},\ }\bibfield  {title} {\bibinfo {title} {\href{https://doi.org/10.1137/S0097539796300921}{Quantum Complexity Theory}},\ }\href {https://doi.org/10.1137/S0097539796300921} {\bibfield  {journal} {\bibinfo  {journal} {SIAM J. Comput.}\ }\textbf {\bibinfo {volume} {26}},\ \bibinfo {pages} {1411–1473} (\bibinfo {year} {1997})}\BibitemShut {NoStop}%
\bibitem [{\citenamefont {Aaronson}(2009)}]{aaronson2009bqp}%
  \BibitemOpen
  \bibfield  {author} {\bibinfo {author} {\bibfnamefont {S.}~\bibnamefont {Aaronson}},\ }\href@noop {} {\bibinfo {title} {\href{https://arxiv.org/abs/0910.4698}{BQP and the Polynomial Hierarchy}}} (\bibinfo {year} {2009}),\ \Eprint {https://arxiv.org/abs/0910.4698} {arXiv:0910.4698 [quant-ph]} \BibitemShut {NoStop}%
\bibitem [{\citenamefont {Edmonds}(1965)}]{edmonds1965}%
  \BibitemOpen
  \bibfield  {author} {\bibinfo {author} {\bibfnamefont {J.}~\bibnamefont {Edmonds}},\ }\bibfield  {title} {\bibinfo {title} {Paths, trees, and flowers},\ }\href {https://doi.org/10.4153/CJM-1965-045-4} {\bibfield  {journal} {\bibinfo  {journal} {Canadian Journal of Mathematics}\ }\textbf {\bibinfo {volume} {17}},\ \bibinfo {pages} {449–467} (\bibinfo {year} {1965})}\BibitemShut {NoStop}%
\bibitem [{\citenamefont {Cobham}(1965)}]{Cobham1965}%
  \BibitemOpen
  \bibfield  {author} {\bibinfo {author} {\bibfnamefont {A.}~\bibnamefont {Cobham}},\ }\bibinfo {title} {The intrinsic computational difficulty of functions},\ in\ \href@noop {} {\emph {\bibinfo {booktitle} {Logic, methodology and philosophy of science}}},\ \bibinfo {editor} {edited by\ \bibinfo {editor} {\bibfnamefont {Y.}~\bibnamefont {Bar{-}Hillel}}}\ (\bibinfo  {publisher} {North-Holland Pub. Co.},\ \bibinfo {year} {1965})\ pp.\ \bibinfo {pages} {24--30}\BibitemShut {NoStop}%
\bibitem [{\citenamefont {Sipser}(1997{\natexlab{a}})}]{sipser13P}%
  \BibitemOpen
  \bibfield  {author} {\bibinfo {author} {\bibfnamefont {M.}~\bibnamefont {Sipser}},\ }\bibinfo {title} {Introduction to the theory of computation}\ (\bibinfo  {publisher} {Course Technology},\ \bibinfo {address} {Boston, MA},\ \bibinfo {year} {1997})\ pp.\ \bibinfo {pages} {234 -- 241},\ \bibinfo {edition} {1st}\ ed.\BibitemShut {Stop}%
\bibitem [{\citenamefont {Sipser}(1997{\natexlab{b}})}]{sipser13NP}%
  \BibitemOpen
  \bibfield  {author} {\bibinfo {author} {\bibfnamefont {M.}~\bibnamefont {Sipser}},\ }\bibinfo {title} {Introduction to the theory of computation}\ (\bibinfo  {publisher} {Course Technology},\ \bibinfo {address} {Boston, MA},\ \bibinfo {year} {1997})\ pp.\ \bibinfo {pages} {241 -- 247},\ \bibinfo {edition} {1st}\ ed.\BibitemShut {Stop}%
\bibitem [{\citenamefont {Sipser}(1997{\natexlab{c}})}]{sipser13NonDet}%
  \BibitemOpen
  \bibfield  {author} {\bibinfo {author} {\bibfnamefont {M.}~\bibnamefont {Sipser}},\ }\bibinfo {title} {Introduction to the theory of computation}\ (\bibinfo  {publisher} {Course Technology},\ \bibinfo {address} {Boston, MA},\ \bibinfo {year} {1997})\ pp.\ \bibinfo {pages} {138 -- 140},\ \bibinfo {edition} {1st}\ ed.\BibitemShut {Stop}%
\bibitem [{\citenamefont {Dempsey}\ and\ \citenamefont {Guinn}(2020)}]{minesweeperPhase}%
  \BibitemOpen
  \bibfield  {author} {\bibinfo {author} {\bibfnamefont {R.}~\bibnamefont {Dempsey}}\ and\ \bibinfo {author} {\bibfnamefont {C.}~\bibnamefont {Guinn}},\ }\bibfield  {title} {\bibinfo {title} {{A Phase Transition in Minesweeper}},\ }in\ \href {https://doi.org/10.4230/LIPIcs.FUN.2021.12} {\emph {\bibinfo {booktitle} {10th International Conference on Fun with Algorithms (FUN 2021)}}},\ \bibinfo {series} {Leibniz International Proceedings in Informatics (LIPIcs)}, Vol.\ \bibinfo {volume} {157},\ \bibinfo {editor} {edited by\ \bibinfo {editor} {\bibfnamefont {M.}~\bibnamefont {Farach-Colton}}, \bibinfo {editor} {\bibfnamefont {G.}~\bibnamefont {Prencipe}},\ and\ \bibinfo {editor} {\bibfnamefont {R.}~\bibnamefont {Uehara}}}\ (\bibinfo  {publisher} {Schloss Dagstuhl -- Leibniz-Zentrum f{\"u}r Informatik},\ \bibinfo {address} {Dagstuhl, Germany},\ \bibinfo {year} {2020})\ pp.\ \bibinfo {pages} {12:1--12:10}\BibitemShut {NoStop}%
\bibitem [{\citenamefont {Achlioptas}\ and\ \citenamefont {Friedgut}(1999)}]{https://doi.org/10.1002/(SICI)1098-2418(1999010)14:1<63::AID-RSA3>3.0.CO;2-7}%
  \BibitemOpen
  \bibfield  {author} {\bibinfo {author} {\bibfnamefont {D.}~\bibnamefont {Achlioptas}}\ and\ \bibinfo {author} {\bibfnamefont {E.}~\bibnamefont {Friedgut}},\ }\bibfield  {title} {\bibinfo {title} {A sharp threshold for k-colorability},\ }\href {https://doi.org/https://doi.org/10.1002/(SICI)1098-2418(1999010)14:1<63::AID-RSA3>3.0.CO;2-7} {\bibfield  {journal} {\bibinfo  {journal} {Random Structures \& Algorithms}\ }\textbf {\bibinfo {volume} {14}},\ \bibinfo {pages} {63} (\bibinfo {year} {1999})}\BibitemShut {NoStop}%
\bibitem [{\citenamefont {Gent}\ and\ \citenamefont {Walsh}(1994)}]{Gent1994TheSP}%
  \BibitemOpen
  \bibfield  {author} {\bibinfo {author} {\bibfnamefont {I.~P.}\ \bibnamefont {Gent}}\ and\ \bibinfo {author} {\bibfnamefont {T.}~\bibnamefont {Walsh}},\ }\bibfield  {title} {\bibinfo {title} {The sat phase transition},\ }in\ \href@noop {} {\emph {\bibinfo {booktitle} {Proceedings of the 11th European Conference on Artificial Intelligence}}},\ \bibinfo {series and number} {ECAI'94}\ (\bibinfo  {publisher} {John Wiley \& Sons, Inc.},\ \bibinfo {address} {USA},\ \bibinfo {year} {1994})\ p.\ \bibinfo {pages} {105–109}\BibitemShut {NoStop}%
\bibitem [{\citenamefont {Stein}\ and\ \citenamefont {Newman}(2013)}]{stein_newman_2013}%
  \BibitemOpen
  \bibfield  {author} {\bibinfo {author} {\bibfnamefont {D.~L.}\ \bibnamefont {Stein}}\ and\ \bibinfo {author} {\bibfnamefont {C.~M.}\ \bibnamefont {Newman}},\ }\href {https://doi.org/10.1515/9781400845637} {\emph {\bibinfo {title} {Spin glasses and complexity}}}\ (\bibinfo  {publisher} {Princeton University Press},\ \bibinfo {year} {2013})\BibitemShut {NoStop}%
\bibitem [{\citenamefont {Galam}\ \emph {et~al.}(1982)\citenamefont {Galam}, \citenamefont {Gefen(Feigenblat)},\ and\ \citenamefont {Shapir}}]{doi:10.1080/0022250X.1982.9989929}%
  \BibitemOpen
  \bibfield  {author} {\bibinfo {author} {\bibfnamefont {S.}~\bibnamefont {Galam}}, \bibinfo {author} {\bibfnamefont {Y.}~\bibnamefont {Gefen(Feigenblat)}},\ and\ \bibinfo {author} {\bibfnamefont {Y.}~\bibnamefont {Shapir}},\ }\bibfield  {title} {\bibinfo {title} {Sociophysics: A new approach of sociological collective behaviour. i. mean‐behaviour description of a strike},\ }\href {https://doi.org/10.1080/0022250X.1982.9989929} {\bibfield  {journal} {\bibinfo  {journal} {The Journal of Mathematical Sociology}\ }\textbf {\bibinfo {volume} {9}},\ \bibinfo {pages} {1} (\bibinfo {year} {1982})}\BibitemShut {NoStop}%
\bibitem [{\citenamefont {Jackson}(2025{\natexlab{a}})}]{jackson2025explainingubiquityphasetransitions}%
  \BibitemOpen
  \bibfield  {author} {\bibinfo {author} {\bibfnamefont {A.}~\bibnamefont {Jackson}},\ }\href {https://arxiv.org/abs/2501.14569} {\bibinfo {title} {Explaining the ubiquity of phase transitions in decision problems}} (\bibinfo {year} {2025}{\natexlab{a}}),\ \Eprint {https://arxiv.org/abs/2501.14569} {arXiv:2501.14569 [cs.CC]} \BibitemShut {NoStop}%
\bibitem [{\citenamefont {Schawe}\ \emph {et~al.}(2019)\citenamefont {Schawe}, \citenamefont {Bleim},\ and\ \citenamefont {Hartmann}}]{101371journalpone0215309}%
  \BibitemOpen
  \bibfield  {author} {\bibinfo {author} {\bibfnamefont {H.}~\bibnamefont {Schawe}}, \bibinfo {author} {\bibfnamefont {R.}~\bibnamefont {Bleim}},\ and\ \bibinfo {author} {\bibfnamefont {A.~K.}\ \bibnamefont {Hartmann}},\ }\bibfield  {title} {\bibinfo {title} {Phase transitions of the typical algorithmic complexity of the random satisfiability problem studied with linear programming},\ }\href {https://doi.org/10.1371/journal.pone.0215309} {\bibfield  {journal} {\bibinfo  {journal} {PLOS ONE}\ }\textbf {\bibinfo {volume} {14}},\ \bibinfo {pages} {1} (\bibinfo {year} {2019})}\BibitemShut {NoStop}%
\bibitem [{\citenamefont {Crawford}\ and\ \citenamefont {Auton}(1996)}]{cross3Sat}%
  \BibitemOpen
  \bibfield  {author} {\bibinfo {author} {\bibfnamefont {J.~M.}\ \bibnamefont {Crawford}}\ and\ \bibinfo {author} {\bibfnamefont {L.~D.}\ \bibnamefont {Auton}},\ }\bibfield  {title} {\bibinfo {title} {Experimental results on the crossover point in random 3-sat},\ }\href {https://doi.org/https://doi.org/10.1016/0004-3702(95)00046-1} {\bibfield  {journal} {\bibinfo  {journal} {Artificial Intelligence}\ }\textbf {\bibinfo {volume} {81}},\ \bibinfo {pages} {31} (\bibinfo {year} {1996})},\ \bibinfo {note} {frontiers in Problem Solving: Phase Transitions and Complexity}\BibitemShut {NoStop}%
\bibitem [{\citenamefont {Bailey}\ \emph {et~al.}(2007)\citenamefont {Bailey}, \citenamefont {Dalmau},\ and\ \citenamefont {Kolaitis}}]{BAILEY20071627}%
  \BibitemOpen
  \bibfield  {author} {\bibinfo {author} {\bibfnamefont {D.~D.}\ \bibnamefont {Bailey}}, \bibinfo {author} {\bibfnamefont {V.}~\bibnamefont {Dalmau}},\ and\ \bibinfo {author} {\bibfnamefont {P.~G.}\ \bibnamefont {Kolaitis}},\ }\bibfield  {title} {\bibinfo {title} {{Phase transitions of PP-complete satisfiability problems}},\ }\href {https://doi.org/https://doi.org/10.1016/j.dam.2006.09.014} {\bibfield  {journal} {\bibinfo  {journal} {Discrete Applied Mathematics}\ }\textbf {\bibinfo {volume} {155}},\ \bibinfo {pages} {1627} (\bibinfo {year} {2007})},\ \bibinfo {note} {{SAT 2001, the Fourth International Symposium on the Theory and Applications of Satisfiability Testing}}\BibitemShut {NoStop}%
\bibitem [{\citenamefont {Hook}\ and\ \citenamefont {Hall}(1991{\natexlab{a}})}]{hook1991solidMagnetic}%
  \BibitemOpen
  \bibfield  {author} {\bibinfo {author} {\bibfnamefont {J.}~\bibnamefont {Hook}}\ and\ \bibinfo {author} {\bibfnamefont {H.}~\bibnamefont {Hall}},\ }\bibinfo {title} {Solid state physics}\ (\bibinfo  {publisher} {Wiley},\ \bibinfo {year} {1991})\ pp.\ \bibinfo {pages} {219--251},\ \bibinfo {edition} {2nd}\ ed.\BibitemShut {Stop}%
\bibitem [{\citenamefont {Buschow}\ and\ \citenamefont {de~Boer}(2003)}]{Buschow2003}%
  \BibitemOpen
  \bibfield  {author} {\bibinfo {author} {\bibfnamefont {K.~H.~J.}\ \bibnamefont {Buschow}}\ and\ \bibinfo {author} {\bibfnamefont {F.~R.}\ \bibnamefont {de~Boer}},\ }\bibinfo {title} {The magnetically ordered state},\ in\ \href {https://doi.org/10.1007/0-306-48408-0_4} {\emph {\bibinfo {booktitle} {Physics of Magnetism and Magnetic Materials}}}\ (\bibinfo  {publisher} {Springer US},\ \bibinfo {address} {Boston, MA},\ \bibinfo {year} {2003})\ pp.\ \bibinfo {pages} {19--42}\BibitemShut {NoStop}%
\bibitem [{\citenamefont {Hook}\ and\ \citenamefont {Hall}(1991{\natexlab{b}})}]{hook1991solid}%
  \BibitemOpen
  \bibfield  {author} {\bibinfo {author} {\bibfnamefont {J.}~\bibnamefont {Hook}}\ and\ \bibinfo {author} {\bibfnamefont {H.}~\bibnamefont {Hall}},\ }\bibinfo {title} {Solid state physics}\ (\bibinfo  {publisher} {Wiley},\ \bibinfo {year} {1991})\ pp.\ \bibinfo {pages} {1--32},\ \bibinfo {edition} {2nd}\ ed.\BibitemShut {Stop}%
\bibitem [{\citenamefont {Hook}\ and\ \citenamefont {Hall}(1991{\natexlab{c}})}]{hook1991solidSuper}%
  \BibitemOpen
  \bibfield  {author} {\bibinfo {author} {\bibfnamefont {J.}~\bibnamefont {Hook}}\ and\ \bibinfo {author} {\bibfnamefont {H.}~\bibnamefont {Hall}},\ }\bibinfo {title} {Solid state physics}\ (\bibinfo  {publisher} {Wiley},\ \bibinfo {year} {1991})\ pp.\ \bibinfo {pages} {278--314},\ \bibinfo {edition} {2nd}\ ed.\BibitemShut {Stop}%
\bibitem [{Note1()}]{Note1}%
  \BibitemOpen
  \bibinfo {note} {Or at least does beyond a short distance from the threshold value}\BibitemShut {NoStop}%
\bibitem [{Note2()}]{Note2}%
  \BibitemOpen
  \bibinfo {note} {In fact, throughout this paper I will regularly use this function and define it polymorphically to return the length of a word regardless of the alphabet the word is formed from.}\BibitemShut {Stop}%
\bibitem [{\citenamefont {Faragó}\ and\ \citenamefont {Xu}(2020)}]{farago2016roughly}%
  \BibitemOpen
  \bibfield  {author} {\bibinfo {author} {\bibfnamefont {A.}~\bibnamefont {Faragó}}\ and\ \bibinfo {author} {\bibfnamefont {R.}~\bibnamefont {Xu}},\ }\href {https://doi.org/https://doi.org/10.1016/j.tcs.2020.03.012} {\bibinfo {title} {A new algorithm design technique for hard problems}} (\bibinfo {year} {2020})\BibitemShut {NoStop}%
\bibitem [{\citenamefont {Cai}(2003)}]{sparseCitation}%
  \BibitemOpen
  \bibfield  {author} {\bibinfo {author} {\bibfnamefont {J.-Y.}\ \bibnamefont {Cai}},\ }\href@noop {} {\bibinfo {title} {Cs 810: Introduction to complexity theory. lecture 11: P=poly, sparse sets, and mahaney's theorem.}} (\bibinfo {year} {2003})\BibitemShut {NoStop}%
\bibitem [{Note3()}]{Note3}%
  \BibitemOpen
  \bibinfo {note} {The complexity class corresponding to practical quantum computing~\cite {Arora_Barak_2009_201-236}}\BibitemShut {NoStop}%
\bibitem [{\citenamefont {Jackson}(2024)}]{jackson2024extensivelypbiimmunepromisebqpcompletelanguages}%
  \BibitemOpen
  \bibfield  {author} {\bibinfo {author} {\bibfnamefont {A.}~\bibnamefont {Jackson}},\ }\href {https://arxiv.org/abs/2406.16764} {\bibinfo {title} {Extensively not p-bi-immune promisebqp-complete languages}} (\bibinfo {year} {2024}),\ \Eprint {https://arxiv.org/abs/2406.16764} {arXiv:2406.16764 [cs.CC]} \BibitemShut {NoStop}%
\bibitem [{\citenamefont {Barz}\ \emph {et~al.}(2013)\citenamefont {Barz}, \citenamefont {Fitzsimons}, \citenamefont {Kashefi},\ and\ \citenamefont {Walther}}]{Barz2013}%
  \BibitemOpen
  \bibfield  {author} {\bibinfo {author} {\bibfnamefont {S.}~\bibnamefont {Barz}}, \bibinfo {author} {\bibfnamefont {J.~F.}\ \bibnamefont {Fitzsimons}}, \bibinfo {author} {\bibfnamefont {E.}~\bibnamefont {Kashefi}},\ and\ \bibinfo {author} {\bibfnamefont {P.}~\bibnamefont {Walther}},\ }\bibfield  {title} {\bibinfo {title} {Experimental verification of quantum computation},\ }\href {https://doi.org/10.1038/nphys2763} {\bibfield  {journal} {\bibinfo  {journal} {Nature Physics}\ }\textbf {\bibinfo {volume} {9}},\ \bibinfo {pages} {727} (\bibinfo {year} {2013})}\BibitemShut {NoStop}%
\bibitem [{\citenamefont {McKague}(2016)}]{v012a003}%
  \BibitemOpen
  \bibfield  {author} {\bibinfo {author} {\bibfnamefont {M.}~\bibnamefont {McKague}},\ }\bibfield  {title} {\bibinfo {title} {Interactive proofs for $\mathsf{BQP}$ via self-tested graph states},\ }\href {https://doi.org/10.4086/toc.2016.v012a003} {\bibfield  {journal} {\bibinfo  {journal} {Theory of Computing}\ }\textbf {\bibinfo {volume} {12}},\ \bibinfo {pages} {1} (\bibinfo {year} {2016})}\BibitemShut {NoStop}%
\bibitem [{\citenamefont {Hangleiter}\ \emph {et~al.}(2017)\citenamefont {Hangleiter}, \citenamefont {Kliesch}, \citenamefont {Schwarz},\ and\ \citenamefont {Eisert}}]{Hangleiter_2017}%
  \BibitemOpen
  \bibfield  {author} {\bibinfo {author} {\bibfnamefont {D.}~\bibnamefont {Hangleiter}}, \bibinfo {author} {\bibfnamefont {M.}~\bibnamefont {Kliesch}}, \bibinfo {author} {\bibfnamefont {M.}~\bibnamefont {Schwarz}},\ and\ \bibinfo {author} {\bibfnamefont {J.}~\bibnamefont {Eisert}},\ }\bibfield  {title} {\bibinfo {title} {Direct certification of a class of quantum simulations},\ }\href {https://doi.org/10.1088/2058-9565/2/1/015004} {\bibfield  {journal} {\bibinfo  {journal} {Quantum Science and Technology}\ }\textbf {\bibinfo {volume} {2}},\ \bibinfo {pages} {015004} (\bibinfo {year} {2017})}\BibitemShut {NoStop}%
\bibitem [{\citenamefont {Kashefi}\ and\ \citenamefont {Wallden}(2017)}]{Kashefi_2017}%
  \BibitemOpen
  \bibfield  {author} {\bibinfo {author} {\bibfnamefont {E.}~\bibnamefont {Kashefi}}\ and\ \bibinfo {author} {\bibfnamefont {P.}~\bibnamefont {Wallden}},\ }\bibfield  {title} {\bibinfo {title} {Optimised resource construction for verifiable quantum computation},\ }\href {https://doi.org/10.1088/1751-8121/aa5dac} {\bibfield  {journal} {\bibinfo  {journal} {Journal of Physics A: Mathematical and Theoretical}\ }\textbf {\bibinfo {volume} {50}},\ \bibinfo {pages} {145306} (\bibinfo {year} {2017})}\BibitemShut {NoStop}%
\bibitem [{\citenamefont {Gheorghiu}\ \emph {et~al.}(2018)\citenamefont {Gheorghiu}, \citenamefont {Kapourniotis},\ and\ \citenamefont {Kashefi}}]{Gheorghiu_2018}%
  \BibitemOpen
  \bibfield  {author} {\bibinfo {author} {\bibfnamefont {A.}~\bibnamefont {Gheorghiu}}, \bibinfo {author} {\bibfnamefont {T.}~\bibnamefont {Kapourniotis}},\ and\ \bibinfo {author} {\bibfnamefont {E.}~\bibnamefont {Kashefi}},\ }\bibfield  {title} {\bibinfo {title} {Verification of quantum computation: An overview of existing approaches},\ }\href {https://doi.org/10.1007/s00224-018-9872-3} {\bibfield  {journal} {\bibinfo  {journal} {Theory of Computing Systems}\ }\textbf {\bibinfo {volume} {63}},\ \bibinfo {pages} {715–808} (\bibinfo {year} {2018})}\BibitemShut {NoStop}%
\bibitem [{\citenamefont {Ferracin}\ \emph {et~al.}(2019)\citenamefont {Ferracin}, \citenamefont {Kapourniotis},\ and\ \citenamefont {Datta}}]{Ferracin_2019}%
  \BibitemOpen
  \bibfield  {author} {\bibinfo {author} {\bibfnamefont {S.}~\bibnamefont {Ferracin}}, \bibinfo {author} {\bibfnamefont {T.}~\bibnamefont {Kapourniotis}},\ and\ \bibinfo {author} {\bibfnamefont {A.}~\bibnamefont {Datta}},\ }\bibfield  {title} {\bibinfo {title} {Accrediting outputs of noisy intermediate-scale quantum computing devices},\ }\href {https://doi.org/10.1088/1367-2630/ab4fd6} {\bibfield  {journal} {\bibinfo  {journal} {New Journal of Physics}\ }\textbf {\bibinfo {volume} {21}},\ \bibinfo {pages} {113038} (\bibinfo {year} {2019})}\BibitemShut {NoStop}%
\bibitem [{\citenamefont {Markham}\ and\ \citenamefont {Krause}(2020)}]{Markham_2020}%
  \BibitemOpen
  \bibfield  {author} {\bibinfo {author} {\bibfnamefont {D.}~\bibnamefont {Markham}}\ and\ \bibinfo {author} {\bibfnamefont {A.}~\bibnamefont {Krause}},\ }\bibfield  {title} {\bibinfo {title} {A simple protocol for certifying graph states and applications in quantum networks},\ }\href {https://doi.org/10.3390/cryptography4010003} {\bibfield  {journal} {\bibinfo  {journal} {Cryptography}\ }\textbf {\bibinfo {volume} {4}},\ \bibinfo {pages} {3} (\bibinfo {year} {2020})}\BibitemShut {NoStop}%
\bibitem [{\citenamefont {Ferracin}\ \emph {et~al.}(2021)\citenamefont {Ferracin}, \citenamefont {Merkel}, \citenamefont {McKay},\ and\ \citenamefont {Datta}}]{Ferracin_2021}%
  \BibitemOpen
  \bibfield  {author} {\bibinfo {author} {\bibfnamefont {S.}~\bibnamefont {Ferracin}}, \bibinfo {author} {\bibfnamefont {S.~T.}\ \bibnamefont {Merkel}}, \bibinfo {author} {\bibfnamefont {D.}~\bibnamefont {McKay}},\ and\ \bibinfo {author} {\bibfnamefont {A.}~\bibnamefont {Datta}},\ }\bibfield  {title} {\bibinfo {title} {Experimental accreditation of outputs of noisy quantum computers},\ }\bibfield  {journal} {\bibinfo  {journal} {Physical Review A}\ }\textbf {\bibinfo {volume} {104}},\ \href {https://doi.org/10.1103/physreva.104.042603} {10.1103/physreva.104.042603} (\bibinfo {year} {2021})\BibitemShut {NoStop}%
\bibitem [{\citenamefont {Jackson}(2025{\natexlab{b}})}]{jackson2025accreditationlimitedadversarialnoise}%
  \BibitemOpen
  \bibfield  {author} {\bibinfo {author} {\bibfnamefont {A.}~\bibnamefont {Jackson}},\ }\href {https://arxiv.org/abs/2409.03995} {\bibinfo {title} {Accreditation against limited adversarial noise}} (\bibinfo {year} {2025}{\natexlab{b}}),\ \Eprint {https://arxiv.org/abs/2409.03995} {arXiv:2409.03995 [quant-ph]} \BibitemShut {NoStop}%
\bibitem [{\citenamefont {Shaffer}\ \emph {et~al.}(2021)\citenamefont {Shaffer}, \citenamefont {Megidish}, \citenamefont {Broz}, \citenamefont {Chen},\ and\ \citenamefont {H{\"a}ffner}}]{Shaffer2021}%
  \BibitemOpen
  \bibfield  {author} {\bibinfo {author} {\bibfnamefont {R.}~\bibnamefont {Shaffer}}, \bibinfo {author} {\bibfnamefont {E.}~\bibnamefont {Megidish}}, \bibinfo {author} {\bibfnamefont {J.}~\bibnamefont {Broz}}, \bibinfo {author} {\bibfnamefont {W.-T.}\ \bibnamefont {Chen}},\ and\ \bibinfo {author} {\bibfnamefont {H.}~\bibnamefont {H{\"a}ffner}},\ }\bibfield  {title} {\bibinfo {title} {Practical verification protocols for analog quantum simulators},\ }\href {https://doi.org/10.1038/s41534-021-00380-8} {\bibfield  {journal} {\bibinfo  {journal} {npj Quantum Information}\ }\textbf {\bibinfo {volume} {7}},\ \bibinfo {pages} {46} (\bibinfo {year} {2021})}\BibitemShut {NoStop}%
\bibitem [{\citenamefont {Jackson}\ \emph {et~al.}(2024)\citenamefont {Jackson}, \citenamefont {Kapourniotis},\ and\ \citenamefont {Datta}}]{doi:10.1073/pnas.2309627121}%
  \BibitemOpen
  \bibfield  {author} {\bibinfo {author} {\bibfnamefont {A.}~\bibnamefont {Jackson}}, \bibinfo {author} {\bibfnamefont {T.}~\bibnamefont {Kapourniotis}},\ and\ \bibinfo {author} {\bibfnamefont {A.}~\bibnamefont {Datta}},\ }\bibfield  {title} {\bibinfo {title} {Accreditation of analogue quantum simulators},\ }\href {https://doi.org/10.1073/pnas.2309627121} {\bibfield  {journal} {\bibinfo  {journal} {Proceedings of the National Academy of Sciences}\ }\textbf {\bibinfo {volume} {121}},\ \bibinfo {pages} {e2309627121} (\bibinfo {year} {2024})}\BibitemShut {NoStop}%
\bibitem [{\citenamefont {Jackson}\ and\ \citenamefont {Datta}(2025)}]{jackson2025improvedaccreditationanaloguequantum}%
  \BibitemOpen
  \bibfield  {author} {\bibinfo {author} {\bibfnamefont {A.}~\bibnamefont {Jackson}}\ and\ \bibinfo {author} {\bibfnamefont {A.}~\bibnamefont {Datta}},\ }\href {https://arxiv.org/abs/2502.06463} {\bibinfo {title} {Improved accreditation of analogue quantum simulation and establishing quantum advantage}} (\bibinfo {year} {2025}),\ \Eprint {https://arxiv.org/abs/2502.06463} {arXiv:2502.06463 [quant-ph]} \BibitemShut {NoStop}%
\bibitem [{\citenamefont {Arora}\ and\ \citenamefont {Barak}(2009{\natexlab{b}})}]{Arora_Barak_2009_201-236}%
  \BibitemOpen
  \bibfield  {author} {\bibinfo {author} {\bibfnamefont {S.}~\bibnamefont {Arora}}\ and\ \bibinfo {author} {\bibfnamefont {B.}~\bibnamefont {Barak}},\ }\href@noop {} {\emph {\bibinfo {title} {Computational Complexity: A Modern Approach}}}\ (\bibinfo  {publisher} {Cambridge University Press},\ \bibinfo {year} {2009})\BibitemShut {NoStop}%
\bibitem [{Note4()}]{Note4}%
  \BibitemOpen
  \bibinfo {note} {$\protect \big [ \cdot \protect \big ]_j$ is also defined polymorphically to function correctly regardless of the alphabet in use.}\BibitemShut {Stop}%
\bibitem [{Note5()}]{Note5}%
  \BibitemOpen
  \bibinfo {note} {Note that this formula, in Eqn.~\ref {eqn:alphaDefiniEquation}, is zero-indexing.}\BibitemShut {Stop}%
\bibitem [{Note6()}]{Note6}%
  \BibitemOpen
  \bibinfo {note} {I note that $\vert \alpha _{\Pi }(n) \vert $ can always be efficiently calculated without using Eqn.~\ref {eqn:alphaDefiniEquation}.}\BibitemShut {Stop}%
\bibitem [{Note7()}]{Note7}%
  \BibitemOpen
  \bibinfo {note} {I note that $\phi ^{-1}: \Sigma ^* \DOTSB \protect \relbar \protect \joinrel \rightarrow \Sigma ^*$ must exist and be a bijection as $\phi $ is due to $\xi $, $\xi ^{-1}$, and $\xi ^{-1} \circ \phi \circ \xi $ being bijections}\BibitemShut {NoStop}%
\bibitem [{Note8()}]{Note8}%
  \BibitemOpen
  \bibinfo {note} {This can be shown more formally, but Fig.~\ref {fig:RelatingPhi} suffices as the key point is the symmetry of the relationship.}\BibitemShut {Stop}%
\bibitem [{Note9()}]{Note9}%
  \BibitemOpen
  \bibinfo {note} {And hence so does $\xi ^{-1}: \Pi \DOTSB \protect \relbar \protect \joinrel \rightarrow \Sigma ^*$}\BibitemShut {NoStop}%
\bibitem [{Note10()}]{Note10}%
  \BibitemOpen
  \bibinfo {note} {I note that the summation over multiple $\protect \mathcal {B}_j^{\phi }$ is required because the set of words of $\Pi ^*$ of a given length is larger than the set of words of $\Sigma ^*$ of the same length (as $\vert \Sigma \vert < \vert \Pi \vert $).}\BibitemShut {Stop}%
\bibitem [{Note11()}]{Note11}%
  \BibitemOpen
  \bibinfo {note} {Also note that the argument above goes through if reversed.}\BibitemShut {Stop}%
\bibitem [{Note12()}]{Note12}%
  \BibitemOpen
  \bibinfo {note} {By which I mean there are not elements not in that set between elements of that set.}\BibitemShut {Stop}%
\bibitem [{Note13()}]{Note13}%
  \BibitemOpen
  \bibinfo {note} {As they are designed to encode exactly the same numbers before and after $\xi $ or $\xi ^{-1}$ act on them, so $\xi $ and $\xi ^{-1}$ preserve the ordering inherited from the encoding of the numbers.}\BibitemShut {Stop}%
\bibitem [{\citenamefont {Antonopoulos}()}]{sparseAndPaddaple}%
  \BibitemOpen
  \bibfield  {author} {\bibinfo {author} {\bibfnamefont {A.}~\bibnamefont {Antonopoulos}},\ }\href {https://courses.corelab.ntua.gr/pluginfile.php/4790/mod_folder/content/0/CC_200121_handouts.pdf?forcedownload=1} {\bibinfo {title} {Computational complexity graduate course}}\BibitemShut {NoStop}%
\end{thebibliography}%

\newpage
\onecolumngrid
\appendix
\section{Alternative Expression of Being Not-Anywhere-Exponentially-Unbalanced}
Def.~\ref{PhiBalanced}, the definition of being not-anywhere-exponentially-unbalanced used herein and in Ref.~\cite{jackson2025explainingubiquityphasetransitions}, may be re-expressed as an equivalent, more mathematical, expression. This is given below as Def~\ref{def:NewNAEU}.
\begin{definition}
    \label{def:NewNAEU}
    A language, $\mathcal{L} \subseteq \Sigma^*$, is \underline{not-anywhere-exponentially-unbalanced} if and only if:
    \begin{align}
        \dfrac{\big \vert \big \{x \in \mathcal{B}_n^{\phi} \text{ } \vert \text{ } \mathcal{P}(x) = \text{Accept} \text{ } \big\} \cap \mathcal{L} \big \vert}{\big \vert \mathcal{B}_n^{\phi} \big \vert},
         \dfrac{\big \vert \big \{x \in \mathcal{B}_n^{\phi} \text{ } \vert \text{ } \mathcal{P}(x) = \text{Reject} \text{ } \big\} \cap \mathcal{L}^c \big \vert}{\big \vert \mathcal{B}_n^{\phi} \big \vert}
         \geq \big( \textit{Poly}(x) \big)^{-1},
    \end{align}
    where $\mathcal{P}: \Sigma^* \longrightarrow \{\text{Accept}, \text{Reject}, \perp\}$  and $\phi: \Sigma^* \longrightarrow \Sigma^*$ are as in Def.~\ref{def:RoughP}, and $\mathcal{L}^c = \big\{ x \in \Sigma^* \text{ }\vert  \text{ } x \not \in \mathcal{L} \big\}$; for some polynomial, $\textit{Poly}: \mathbb{N}_0 \longrightarrow \mathbb{R}$, such that: $\forall n \in \mathbb{N}_0$,
    \begin{align}
    \dfrac{\partial}{\partial x} \bigg(\textit{Poly}(x) \cdot \big( 1 / \sqrt{2} \big)^x \bigg) \bigg \vert_{x = n} \leq 0.
    \end{align}
\end{definition}
For convenience of comparison, I also restate Def.~\ref{PhiBalanced} below.
\begin{definitionRestate}\ref{PhiBalanced} (Restated)\\
    \label{PhiBalancedRestate}
    A language is \underline{not-anywhere-exponentially-unbalanced} if there exists some polynomial, $\textit{Poly}: \mathbb{N} \longrightarrow \mathbb{R}$, such that $\forall n \in \mathbb{N}$, neither the fraction of $\mathcal{B}_n^{\phi}$ that is in the language (and $\mathcal{P}$ decides correctly) nor the fraction not in the language (and $\mathcal{P}$ decides correctly) are less than $\big( \textit{Poly}(n) \big)^{-1}$, and $\textit{Poly}(n) \cdot \big( 1 / \sqrt{2} \big)^n$ is monotonically decreasing.
\end{definitionRestate}
For completeness, I state the below Lemma~\ref{lem:defsEquiv} without proof (as Def.~\ref{def:NewNAEU} is just Def.~\ref{PhiBalanced} expressed more mathematically).
\begin{lemma}
    \label{lem:defsEquiv}
    Def.~\ref{PhiBalanced} and Def.~\ref{def:NewNAEU} are equivalent.
\end{lemma}

\section{Auxiliary Lemmas for the Proof of Theorem~\ref{lem:MainLemma}}
\label{app:LemsForMainLemma}
This appendix states and proves Lemma~\ref{lem:PIsoExists} and Lemma~\ref{lem:PisoImpliesPhase}, both of which are required by the proof of Theorem~\ref{lem:MainLemma} in Sec.~\ref{sec:proofOfMainLemma}.
There are also several other lemmas, in this appendix, to assist in proving these lemmas.
\subsection{Proving Lemma~\ref{lem:PIsoExists}}
\begin{lemma}
    \label{lem:BasicIsomorpLemma}
    Any language over an alphabet $\Sigma = \{ 1, 2, ..., \vert \Sigma \vert \}$ is preserving-P-isomorphic to a language over the alphabet $\Pi = \{ 1, 2, ..., \vert \Sigma \vert,  \vert \Sigma \vert + 1 \}$.
\end{lemma}
\begin{proof}
    Let $\mathcal{L} \subseteq \Sigma^*$ be a language. I then aim to construct the required language over $\Pi$ that $\mathcal{L}$ is preserving-P-isomorphic to.
    This is accomplished by defining the language $\mathcal{L}$ is preserving-P-isomorphic to, $\mathcal{H}$, by:
    \begin{align}
        \label{eqn:preservinginAppendixA}
        \xi(x) \in \mathcal{H} \iff x \in \mathcal{L},
    \end{align}
    where $\xi: \Sigma^* \longrightarrow \Pi^*$ is some to-be-defined polynomial-time computable P-isomorphism. I.e. Eqn.~\ref{eqn:preservinginAppendixA} is tautologically true.
    
    Note that this definition automatically makes $\xi$ a preserving-P-isomorphism (as in Def.~\ref{def:preservingPIso}) between $\mathcal{L}$ and $\mathcal{H}$.
    
    Properly defining $\xi$ first requires defining two functions that are composed to form it, in Def.~\ref{def:OmegaSigma} and Def.~\ref{def:alphaPi}.
    \begin{definition}
        \label{def:OmegaSigma}
        Define $\theta_{\Sigma}: \Sigma^* \longrightarrow \mathbb{N}_0$ by: $\forall x \in \Sigma^*$,
        \begin{align}
            \theta_{\Sigma}(x)
            &=
            \sum^{\vert x \vert - 1}_{j = 0} \bigg( \vert \Sigma \vert^j \big[ x \big]_j \bigg),
        \end{align}
        where $\big[ \cdot \big]_j: \Sigma^* \longrightarrow \mathbb{N}^{\leq \vert \Sigma \vert}$ returns the symbol with index $j \in \mathbb{N}$ (using zero-indexing) in its argument word~\footnote{$\big[ \cdot \big]_j$ is also defined polymorphically to function correctly regardless of the alphabet in use.}. I.e. $\theta_{\Sigma}$ interprets any string in $\Sigma^*$ as a little-endian expression of an integer, in base $\vert \Sigma \vert$, and returns that integer.
    \end{definition}
    \begin{definition}
        \label{def:alphaPi}
        Define~\footnote{Note that this formula, in Eqn.~\ref{eqn:alphaDefiniEquation}, is zero-indexing.} $\alpha_{\Pi}: \mathbb{N}_0 \longrightarrow \Pi^*$ by: $\forall n \in \mathbb{N}_0$, $\forall j \in \mathbb{N}^{\leq \vert \alpha_{\Pi}(n) \vert - 1}$ \footnote{I note that $\vert \alpha_{\Pi}(n) \vert$ can always be efficiently calculated without using Eqn.~\ref{eqn:alphaDefiniEquation}.},
        \begin{align}
            \label{eqn:alphaDefiniEquation}
            \big[ \alpha_{\Pi}(n) \big]_{j}
            &=
            \bigg \lfloor \dfrac{n}{\vert \Pi \vert^j}  \bigg \rfloor \text{ mod } \vert \Pi \vert.
        \end{align}
    \end{definition}
    I then define my proposed preserving-P-isomorphism, $\xi: \Sigma^* \longrightarrow \Pi^*$, as in Def.~\ref{def:DefingXi}.
    \begin{definition}
    \label{def:DefingXi}
        Define the preserving-P-isomorphism, 
        $\forall x \in \Sigma^*$, $\xi: \Sigma^* \longrightarrow \Pi^*$, by:
    \begin{align}
        \xi(x)
        &=
        \alpha_{\Pi} \big( \theta_{\Sigma}(x) \big).
    \end{align}
    As each of its constituent parts are polynomial-time computable, so is $\xi$. 
    \end{definition}
    
    Therefore it only remains to show $\xi: \Sigma^* \longrightarrow \Pi^*$ is an isomorphism. For this, I propose the following candidate for the inverse of $\xi$ (i.e. $\xi^{-1}: \Pi^* \longrightarrow \Sigma^*$): $\forall y \in \Pi^*$,
    \begin{align}
        \xi^{-1}(y) &= \alpha_{\Sigma} \big( \theta_{\Pi}(y) \big).
    \end{align}
    $\xi^{-1}$ is polynomial time computable as, like $\xi$, it is composed of a fixed number of polynomial time functions and inverts $\xi$ as: for any alphabet, $\Xi$, and $\forall z \in \Xi^*$, $\forall n \in \mathbb{N}_0$,
    \begin{align}
        &\forall j \in \mathbb{N}^{\leq \vert \alpha_{\Xi} \big( \theta_{\Xi}(z) \big) \vert - 1}, \big[  \alpha_{\Xi} \big( \theta_{\Xi}(z) \big) \big]_{k}
        =
        \bigg \lfloor \dfrac{\sum^{\vert z \vert - 1}_{j = 0} \big( \vert \Xi \vert^j \big[ z \big]_j \big)}{\vert \Xi \vert^k}  \bigg \rfloor \text{ mod } \vert \Xi \vert
        =
        \bigg \lfloor \sum^{\vert z \vert - 1}_{j = 0} \big( \vert \Xi \vert^{j-k} \big[ z \big]_j \big)  \bigg \rfloor \text{ mod } \vert \Xi \vert
         =
         \sum^{\vert z \vert -1}_{j = k} \big( \vert \Xi \vert^{j-k} \big[ z \big]_j \big) \text{ mod } \vert \Xi \vert
         =
        \big[ z \big]_{k}
        \label{Eqn:FirstToRedInEqn}\\
        &\text{ and; as } \forall n \in \mathbb{N}_0, \exists q_n \in \Xi^* \text{ such that } \theta_{\Xi}(q_n) = n; \hspace{0.5cm} \theta_{\Xi} \big( \alpha_{\Xi}\big (n \big) \big)
        =
        \theta_{\Xi} \big( \alpha_{\Xi}\big ( \theta_{\Xi}(q_n) \big) \big)
        =
        \theta_{\Xi}(q_n)
        = n
        \label{Eqn:secondToRedInEqn},
    \end{align}
    where the second equality in Eqn.~\ref{Eqn:secondToRedInEqn} follows from Eqn.~\ref{Eqn:FirstToRedInEqn}. Therefore, I conclude this proof by showing: $\forall x \in \Sigma^*$, $\forall y \in \Pi^*$,
    \begin{align}
        \xi^{-1} (\xi(x))
        =
        \alpha_{\Sigma} \big( \theta_{\Pi}( \alpha_{\Pi} \big( \theta_{\Sigma}(x) \big)) \big)
        =
        \alpha_{\Sigma} \big( \theta_{\Sigma}(x) \big)
        =
        x
       \hspace{0.5cm} \text{ and } \hspace{0.5cm}
        \xi(\xi^{-1}(y)) 
        &=
        \alpha_{\Pi} \big( \theta_{\Sigma}( \alpha_{\Sigma} \big( \theta_{\Pi}(y) \big) ) \big)
        = 
        \alpha_{\Pi} \big( \theta_{\Pi}(y) \big)
        =
        y.
    \end{align}
    Hence $\xi$ is an isomorphism. Therefore, also using Eqn.~\ref{eqn:preservinginAppendixA} and that $\xi$ can be impliemented efficiently, I conclude that  $\xi$ is a preserving-P-isomorphism.
\end{proof}
The situation established in Lemma~\ref{lem:BasicIsomorpLemma} is depicted in Fig.~\ref{fig:Lemma3Diagram}.
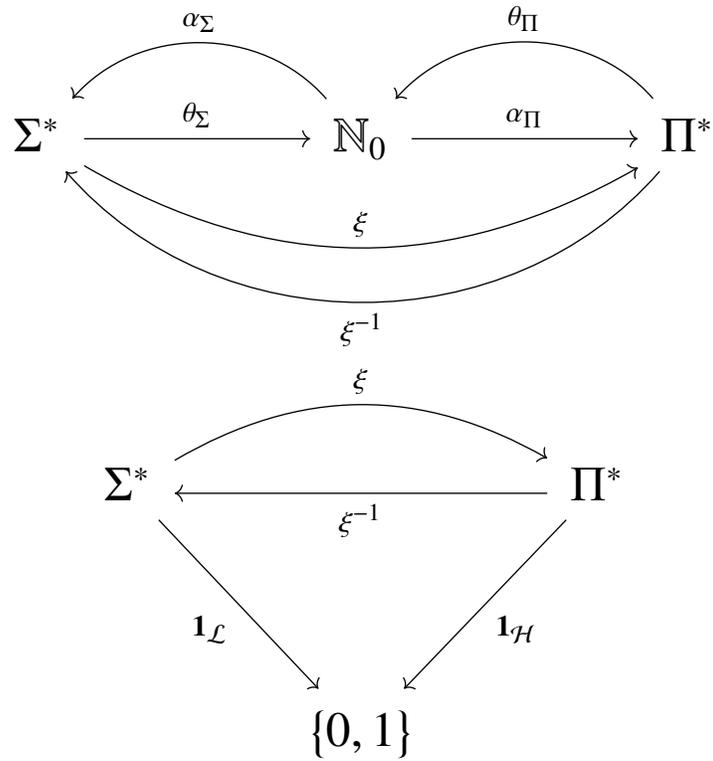
\begin{figure}[h!]
    \centering
    \begin{tikzcd}[scale cd=2.3, row sep=huge,column sep=huge]
\Sigma^* \arrow[rr, "\theta_{\Sigma}"] \arrow[rrrr, "\xi", bend right] &  & \mathbb{N}_0 \arrow[rr, "\alpha_{\Pi}"] \arrow[ll, "\alpha_{\Sigma}"', bend right=49] &  & \Pi^* \arrow[ll, "\theta_{\Pi}"', bend right=49] \arrow[llll, "\xi^{-1}", bend left=49]
\end{tikzcd}
    
       \begin{tikzcd}[scale cd=2.3, row sep=huge,column sep=huge] 
\Sigma^* \arrow[rr, "\xi", bend left] \arrow[rdd, "\mathbf{1}_{\mathcal{L}}"'] &                       & \Pi^* \arrow[ll, "\xi^{-1}", shift left] \arrow[ldd, "\mathbf{1}_{\mathcal{H}}"] \\
                                                                               &                       &                                                                                  \\
                                                                               & {\big \{ 0, 1 \big\}} &                                                                                 
\end{tikzcd}
    \caption{Diagrams of the situation constructed in Lemma~\ref{lem:BasicIsomorpLemma}, where $\mathbf{1}_{\mathcal{L}}: \Sigma^* \longrightarrow \{ 0, 1\}$ is the indicator function of $\mathcal{L}$, defined as, $\forall x \in \Sigma^*$, $\mathbf{1}_{\mathcal{L}}(x) = \begin{cases}
        1 &\text{ if } x \in \mathcal{L}\\
        0 &\text{ if } x \not \in \mathcal{L}
    \end{cases}$ and $\mathbf{1}_{\mathcal{H}}: \Sigma^* \longrightarrow \{ 0, 1\}$ is the indicator function of $\mathcal{H}$, defined as, $\forall x \in \Pi^*$, $\mathbf{1}_{\mathcal{H}}(x) = \begin{cases}
        1 &\text{ if } x \in \mathcal{H}\\
        0 &\text{ if } x \not \in \mathcal{H}
    \end{cases}$.\\ 
    The upper diagram depicts the construction of $\xi$ and $\xi^{-1}$, while the lower diagram depicts the functioning of $\xi$ and $\xi^{-1}$ to preserve membership of the respective languages.
    Lemma~\ref{lem:BasicIsomorpLemma} is equivalent to saying that a $\mathcal{H} \subseteq \Pi^*$ exists for any $\mathcal{L} \subseteq \Sigma^*$ such that the lower diagram commutes.}
    \label{fig:Lemma3Diagram}
\end{figure}

\begin{lemma}
    \label{lem:PIsoExists}
    Any adequately-balanced language over an odd-sized alphabet is preserving-P-isomorphic to a paddable not-anywhere-exponentially-unbalanced language over an even-sized alphabet.
\end{lemma}
\begin{proof}
Let $\Sigma$ be an alphabet such that $\vert \Sigma \vert$ is odd and let $\mathcal{L} \subseteq \Sigma^*$ be an adequately-balanced language. I then aim to construct the required language (that will be denoted $\mathcal{H}$) over an even-sized alphabet (that will be denoted $\Pi$) that $\mathcal{L}$ is preserving-P-isomorphic to.

Assume, without loss of generality, that $\Sigma = \big \{ 1, 2, ..., \vert \Sigma \vert \big \}$ and define another alphabet, $\Pi$, by:
\begin{align}
    \Pi
    &=
    \Sigma \cup \big \{ \vert \Sigma \vert + 1 \big \}
    =
    \big \{ 1, 2, ..., \vert \Sigma \vert, \vert \Sigma \vert + 1\big \}.
\end{align}
As $\big \vert \Sigma \big \vert $ is, by assumption, odd, $ \big \vert \Pi \big \vert = \big \vert \Sigma \cup \big \{ \Sigma + 1 \big \} \big \vert = \big\vert \Sigma \big \vert + 1$ is even. Due to Lemma~\ref{lem:BasicIsomorpLemma}, a preserving-P-isomorphism, $\xi: \Sigma^* \longrightarrow \Pi^*$, exists between $\mathcal{L}$ and a language, denoted $\mathcal{H}$, over $\Pi$, defined by $ \mathcal{H} = \big \{ \xi(x) \text{ } \vert \text{ } x \in \mathcal{L} \big \}$.

I now aim to show that this preserving-P-isomorphism implies $\mathcal{H} \subseteq \Pi^*$ is also paddable and not-anywhere-exponentially-unbalanced. I address the paddability and not-anywhere-exponentially-unbalanced requirements separately.\\

\noindent \underline{\textbf{$\mathcal{H}$ is paddable}}\\
As $\mathcal{L}$ is assumed to be paddable, padding and decoding functions for $\mathcal{L}$ must exist. Therefore, let:
\begin{enumerate}
    \item $\textit{Pad}_{\mathcal{L}}: \Sigma^* \times \Sigma^* \longrightarrow \Sigma^*$ be the padding function of $\mathcal{L}$. I.e. $\forall x, y, \in \Sigma^*$, $\textit{Pad}_{\mathcal{L}}(x, y) \in \mathcal{L} \iff x \in \mathcal{L}$.
    \item $\textit{Dec}_{\mathcal{L}}: \Sigma^* \longrightarrow \Sigma^*$ be the decoding function of $\mathcal{L}$. I.e. $\forall x, y, \in \Sigma^*$, $\textit{Dec}_{\mathcal{L}} \big( \textit{Pad}_{\mathcal{L}}(x, y) \big) = y$.
\end{enumerate}
Using these padding and decoding functions for $\mathcal{L}$, I propose the following as padding and decoding functions for $\mathcal{H}$: 
\begin{enumerate}
    \item $\textit{Pad}_{\mathcal{H}}: \Pi^* \times \Pi^* \longrightarrow \Pi^*$ is the padding function of $\mathcal{H}$, defined by: $\forall y, z \in \Pi^*$,
    \begin{align}
        \label{eqn:PadHDef}
        \textit{Pad}_{\mathcal{H}} \big( y, z \big)
        &=
        \xi \big[ \textit{Pad}_{\mathcal{L}} \big( \xi^{-1}[y], \xi^{-1}[z] \big) \big].
    \end{align}
    \item $\textit{Dec}_{\mathcal{H}}: \Pi^* \longrightarrow \Pi^*$ be the decoding function of $\mathcal{H}$, defined by: $\forall z \in \Pi^*$,
    \begin{align}
        \label{eqn:DecHDef}
        \textit{Dec}_{\mathcal{H}} \big( z \big)
        &=
        \xi \big[ \textit{Dec}_{\mathcal{L}} \big( \xi^{-1}[z] \big) \big].
    \end{align}
\end{enumerate}
These definitions can also be read off of Fig.~\ref{fig:Lemma4Diagram} by diagram chasing.

I now show that the above proposed padding and decoding functions for $\mathcal{H}$ perform as required. The first property required of the padding function is that it preserves the first argument's membership of $\mathcal{H}$, or lack thereof. To show this consider: $\forall y, z \in \Pi^*$,
\begin{align}
    y \in \mathcal{H}
    &\iff
    \xi^{-1}[y] \in \mathcal{L}, &\text{ as $\xi^{-1}: \Pi^* \longrightarrow \Pi^*$ is a preserving-P-isomorphism}\\
    &\iff
    \forall z \in \Pi^*, \textit{Pad}_{\mathcal{L}} \big( \xi^{-1}[y], \xi^{-1}[z]\big) \in \mathcal{L}, &\text{ by the definition of a padding function}\\ 
    &\iff
     \forall z \in \Pi^*, \xi \big[ \textit{Pad}_{\mathcal{L}} \big( \xi^{-1}[y], \xi^{-1}[z]\big) \big] \in \mathcal{H},
     &\text{ as $\xi: \Pi^* \longrightarrow \Pi^*$ is a preserving-P-isomorphism}\\
     &\iff
     \forall z \in \Pi^*,
     \textit{Pad}_{\mathcal{H}} \big( y, z \big) \in \mathcal{H},  &\text{ by the definition of $\textit{Pad}_{\mathcal{H}}: \Pi^* \times \Pi^* \longrightarrow \Pi^*$}.
\end{align}
I.e. $\textit{Pad}_{\mathcal{H}}$ preserves the membership of $\mathcal{H}$ or $\mathcal{H}^c$ of its first argument.

The other required property of both $\textit{Pad}_{\mathcal{H}}$ and $\textit{Dec}_{\mathcal{H}}$, in tandem, is that $\textit{Dec}_{\mathcal{H}}$ correctly recovers any word in $\Pi^*$ encoded by $\textit{Pad}_{\mathcal{H}}$. This requirement may be expressed alternatively as: $\forall y, z \in \Pi^*$,
\begin{align}
    \textit{Dec}_{\mathcal{H}} \big( \textit{Pad}_{\mathcal{H}} \big(y, z \big) \big) = z.
\end{align}
Using the definitions of $\textit{Pad}_{\mathcal{H}}$ (in Eqn.~\ref{eqn:PadHDef}) and $\textit{Dec}_{\mathcal{H}}$ (in Eqn.~\ref{eqn:DecHDef}), this can be confirmed: $\forall y, z \in \Pi^*$,
\begin{align}
    \textit{Dec}_{\mathcal{H}} \big( \textit{Pad}_{\mathcal{H}} \big(y, z \big) \big) 
    &=
     \xi \big[ \textit{Dec}_{\mathcal{L}} \big( \xi^{-1}\big[\xi \big[ \textit{Pad}_{\mathcal{L}} \big( \xi^{-1}[y], \xi^{-1}[z] \big) \big] \big] \big) \big]
     =
     \xi \big[ \textit{Dec}_{\mathcal{L}} \big(  \textit{Pad}_{\mathcal{L}} \big( \xi^{-1}[y], \xi^{-1}[z] \big) \big) \big]
     =
     \xi \big[ \xi^{-1}[z] \big]
     =
     z.
\end{align}
Similarly to before, the functioning of the padding and decoding functions of $\mathcal{H}$ can also be read off of Fig.~\ref{fig:Lemma4Diagram} by diagram chasing.
\newline
Hence, all required properties of the padding and decoding functions for $\mathcal{H}$ to be paddable are fulfilled.\\

The situation derived in this demonstration of the paddability of $\mathcal{H}$ is depicted in Fig.~\ref{fig:Lemma4Diagram}. In fact, the above demonstration of the paddability of $\mathcal{H}$ implies that the diagram in Fig.~\ref{fig:Lemma4Diagram} commutes.
\newline 

\begin{figure}[h!]
    \centering
    \begin{tikzcd}[scale cd=2.1, row sep=huge,column sep=huge]
                                              &  &                                                                                                                               & {\big\{ 0, 1 \big \}} &                                                                                                                                          &  &                                                   \\
                                              &  &                                                                                                                               &                       &                                                                                                                                          &  &                                                   \\
\Pi^* \arrow[rrruu, "\bold{1}_{\mathcal{H}}"] &  & \Pi^* \arrow[ll, "\textit{Dec}_{\mathcal{H}}"'] \arrow[rr, "\xi^{-1}"] \arrow[ruu, "\bold{1}_{\mathcal{H}}"']                 &                       & \Sigma^* \arrow[rr, "\textit{Dec}_{\mathcal{L}}"] \arrow[ll, "\xi", bend left] \arrow[luu, "\bold{1}_{\mathcal{L}}"]                     &  & \Sigma^* \arrow[llluu, "\bold{1}_{\mathcal{L}}"'] \\
                                              &  &                                                                                                                               &                       &                                                                                                                                          &  &                                                   \\
                                              &  & \Pi^* \times \Pi^* \arrow[rr, "\xi^{-1} \times \xi^{-1}"'] \arrow[uu, "\textit{Pad}_{\mathcal{H}}"] \arrow[lluu, "i^{\pi}_2"] &                       & \Sigma^* \times \Sigma^* \arrow[ll, "\xi \times \xi", bend left] \arrow[uu, "\textit{Pad}_{\mathcal{L}}"'] \arrow[rruu, "i^{\Sigma}_2"'] &  &                                                  
\end{tikzcd}
    \caption{
    Diagram of the situation constructed in Lemma~\ref{lem:PIsoExists}, where $\mathbf{1}_{\mathcal{L}}: \Sigma^* \longrightarrow \{ 0, 1\}$ is the indicator function of $\mathcal{L}$, as in Fig.~\ref{fig:Lemma3Diagram}, and $\mathbf{1}_{\mathcal{H}}: \Sigma^* \longrightarrow \{ 0, 1\}$ is the indicator function of $\mathcal{H}$, as in Fig.~\ref{fig:Lemma3Diagram}. 
    Additionally, $i_2^\Sigma: \Sigma^* \times \Sigma^* \longrightarrow \Sigma^*$ is defined by, $\forall x, y \in \Sigma^*$, $i_2^\Sigma(x,y) = y$, and $i_2^\pi: \Pi^* \times \Pi^* \longrightarrow \Pi^*$ is defined similarly.\\
    Lemma~\ref{lem:PIsoExists} implies this diagram commutes and all mappings -- except the indicator functions -- can be applied in polynomial time.\\
    Note the symmetry of the above diagram, along the center line, between $\Sigma^*$ and $\Pi^*$ and that the left-hand-side of the above diagram follows completely from the right-hand-side, given $\xi$ and $\xi^{-1}$.
    }
    \label{fig:Lemma4Diagram}
\end{figure}
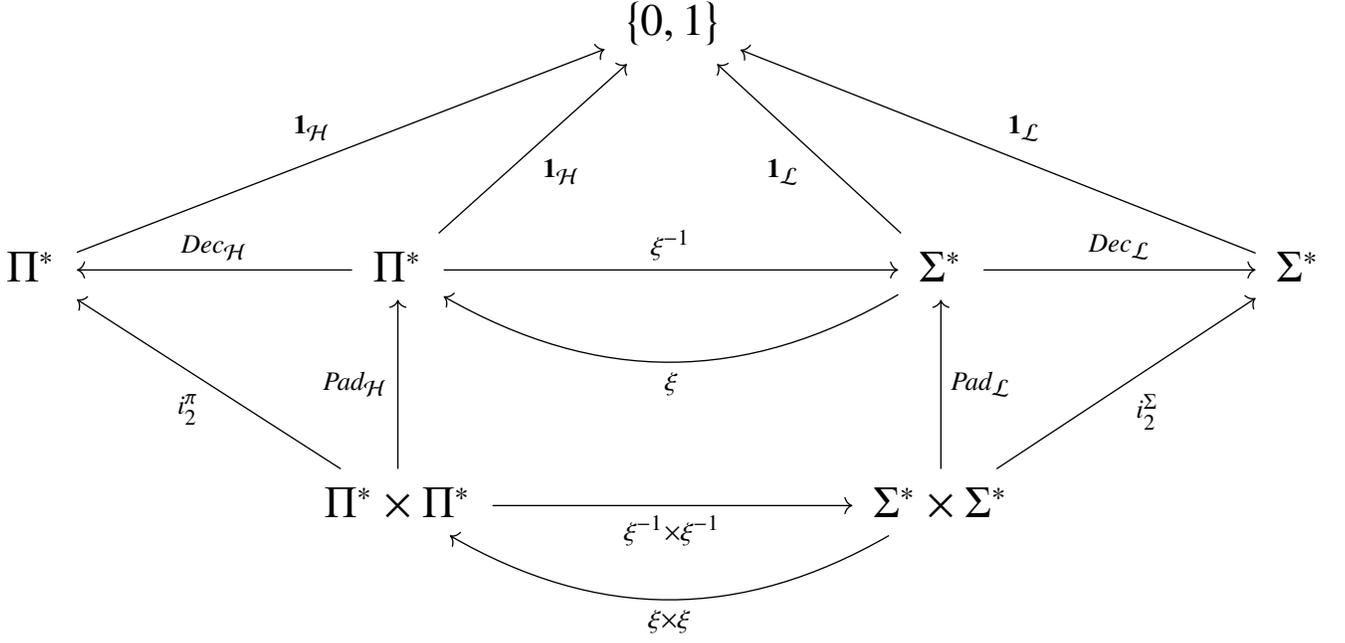

\noindent \underline{\textbf{$\mathcal{H}$ is not-anywhere-exponentially-unbalanced}}\\
I first examine for which function $\psi: \Pi^* \longrightarrow \Pi^*$ I should evaluate the fractions of $\mathcal{B}^{\psi}_n$ both in and not in $\mathcal{H}$ for. I.e. what being not-anywhere-exponentially-unbalanced would really mean for $\mathcal{H}$ and which $\mathcal{B}^{\psi}_n$ it is defined in terms of.

I note that, as $\mathcal{H}$ is has been shown to be paddable (above), Ref.~\cite[Theorem 1]{farago2016roughly} implies $\mathcal{H}$ is in RoughP (as defined in Def.~\ref{def:RoughP}). Therefore, similarly to as in Ref.~\cite{farago2016roughly}, a preserving-P-isomorphism, $\xi^{-1} \circ \phi \circ \xi:$ $\Pi^* \longrightarrow \Pi^*$, as in Def.~\ref{def:RoughP}, exists that maps elements of $\mathcal{H}$ to elements of: 
\begin{align}
    H_{\mathcal{H}}
    &=
    \big \{ x x \text{ } \vert \text{ } x \in \mathcal{H}\big \} \cup \big \{ x \text{ } \vert \text{ } \omega \big( x \big) \text{ is odd}\big \} \subseteq \Pi^*,
\end{align}
(where $\omega: \Pi^* \rightarrow \mathbb{N}_0$ sums the elements of a given string -- assuming they are numerical, which I can assume WLOG) and maps elements of $\mathcal{H}^c$ to elements of $H_{\mathcal{H}}^c$.
This strange notation for the above-mentioned preserving-P-isomorphism is to highlight how $\xi^{-1} \circ \phi \circ \xi:$ $\Pi^* \longrightarrow \Pi^*$ implicitly defines a $\phi:$ $\Sigma^* \longrightarrow \Sigma^*$, as in Fig.~\ref{fig:RelatingPhi} and Fig.~\ref{fig:FinalProofDiagram}, that maps elements of $\mathcal{L}$ to elements of
\begin{align}
    H_{\mathcal{L}}
    &=
    \big \{ x \text{ } \vert \text{ } \phi^{-1} ( x ) \in \mathcal{L} \big \}  \subseteq \Sigma^*,
\end{align}
and maps elements of $\mathcal{L}^c$ to elements of $H_{\mathcal{L}}^c$~\footnote{I note that $\phi^{-1}: \Sigma^* \longrightarrow \Sigma^*$ must exist and be a bijection as $\phi$ is due to $\xi$, $\xi^{-1}$, and $\xi^{-1} \circ \phi \circ \xi$ being bijections}.

To identify the correct $\psi$, to define $\mathcal{H}$ being not-anywhere-exponentially-unbalanced via, consider Fig.~\ref{fig:RelatingPhi} which depicts the -- thus far -- constructed situation, detailing the relationships between various relevant mappings. One of Fig.~\ref{fig:RelatingPhi} lines of symmetry, which runs horizontally, cutting every instance of $\xi$ and $\xi^{-1}$ in Fig.~\ref{fig:RelatingPhi}, shows that $\xi^{-1} \circ \phi \circ \xi$ has the exact same relation to $\textit{Pad}_{\mathcal{H}}$ and $\textit{Dec}_{\mathcal{H}}$ that $\phi$ has to $\textit{Pad}_{\mathcal{L}}$ and $\textit{Dec}_{\mathcal{L}}$~\footnote{This can be shown more formally, but Fig.~\ref{fig:RelatingPhi} suffices as the key point is the symmetry of the relationship.} and, by diagram chasing on Fig.~\ref{fig:RelatingPhi}, $\xi^{-1} \circ \phi \circ \xi$ can be expressed in terms of the padding and decoding functions of $\mathcal{H}$ in the same way that $\phi$ can be expressed in terms of the  padding and decoding functions of $\mathcal{L}$ (directly substituting the padding and decoding functions of the respective languages).

So, I conclude that the correct choice of $\psi$ is $\xi^{-1} \circ \phi \circ \xi$ and therefore the conditions for $\mathcal{H}$ being not-anywhere-exponentially-unbalanced are expressed in terms of $\mathcal{B}^{\xi^{-1} \circ \phi \circ \xi}_n$.

\begin{figure}[h!]
    \begin{tikzcd}[scale cd=1.8, row sep=huge,column sep=huge]
\Pi^* \times \Pi^* \arrow[dd, "\textit{Pad}_{\mathcal{H}}"'] \arrow[rrdd, "i^{\pi}_2"]          &  &                                                                                                                           & {\{0,1\}}  &                                                                                    &  & \Pi^* \times \Pi^* \arrow[lldd, "i^{\pi}_2"'] \arrow[dd, "\textit{Pad}_{H_{\mathcal{H}}}"] \\
                                                                                                &  &                                                                                                                           &            &                                                                                    &  &                                                                                            \\
\Pi^* \arrow[dd, "\xi^{-1}", bend left] \arrow[rr, "\textit{Dec}_{\mathcal{H}}"]                &  & \Pi^* \arrow[dd, "\xi^{-1}", bend left] \arrow[rr, "\xi^{-1} \circ \phi \circ \xi"] \arrow[ruu, "\bold{1}_{\mathcal{H}}"] &            & \Pi^* \arrow[dd, "\xi^{-1}", bend left] \arrow[luu, "\bold{1}_{H_{\mathcal{H}}}"'] &  & \Pi^*  \arrow[dd, "\xi^{-1}", bend left] \arrow[ll, "\textit{Dec}_{H_{\mathcal{H}}}"']     \\
                                                                                                &  &                                                                                                                           &            &                                                                                    &  &                                                                                            \\
\Sigma^* \arrow[uu, "\xi", bend left] \arrow[rr, "\textit{Dec}_{\mathcal{L}}"]                  &  & \Sigma^* \arrow[rr, "\phi"] \arrow[uu, "\xi", bend left] \arrow[rdd, "\bold{1}_{\mathcal{L}}"']                           &            & \Sigma^* \arrow[uu, "\xi", bend left] \arrow[ldd, "\bold{1}_{H_{\mathcal{L}}}"]                  &  & \Sigma^* \arrow[uu, "\xi", bend left] \arrow[ll, "\textit{Dec}_{H_{\mathcal{L}}}"']                      \\
                                                                                                &  &                                                                                                                           &            &                                                                                    &  &                                                                                            \\
\Sigma^* \times \Sigma^* \arrow[uu, "\textit{Pad}_{\mathcal{L}}"] \arrow[rruu, "i_2^{\Sigma}"'] &  &                                                                                                                           & {\{0,1 \}} &                                                                                    &  & \Sigma^* \times \Sigma^* \arrow[uu, "\textit{Pad}_{H_{\mathcal{L}}}"'] \arrow[lluu, "i_2^{\Sigma}"]     
\end{tikzcd}
    \caption{A diagram entirely entailed by the lower half, which follows from both $\mathcal{L}$ and $H$ being paddable and mutually reducible (implying the existence of $\phi$), and $\xi: \Sigma^* \longrightarrow \Pi^*$ being a preserving-P-isomorphism. Note the symmetry between the upper and lower halves. I additionally note that the two instance of $\{ 0, 1 \}$ can be merged into one, keeping all arrows into them the same, and the diagram still commutes (with the horizontal symmetry also maintained).}
    \label{fig:RelatingPhi}
\end{figure}
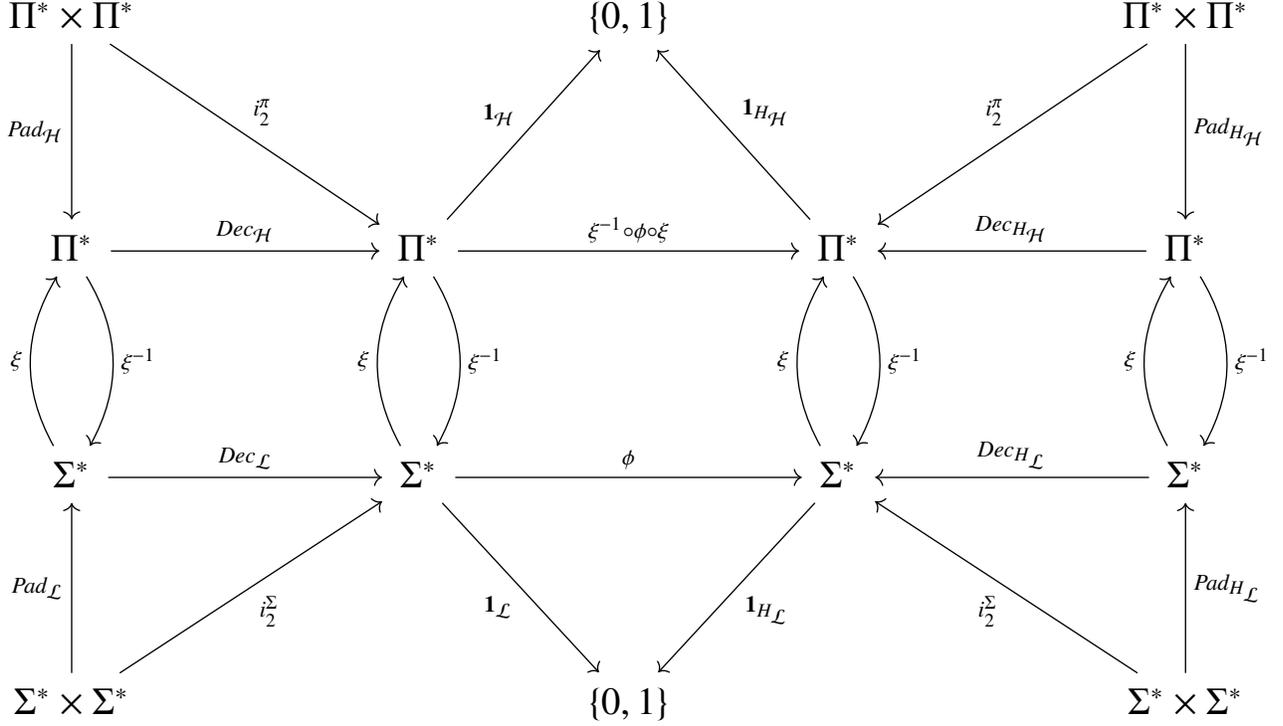

Therefore, to show that $\mathcal{H}$ is not-anywhere-exponentially-unbalanced, the mathematical expression I have to bound from below (by the reciprocal of a  polynomial, $\textit{Poly}: \mathbb{N}_0 \longrightarrow \mathbb{R}$, such that $\textit{Poly}(n) \big( 1 / \sqrt{2} \big)^n$ is monotonically decreasing, as in Def.~\ref{PhiBalanced}), is:
\begin{align}
    \label{eqn:NAEUQuantity}
    \dfrac{\big \vert \mathcal{B}_n^{\xi^{-1} \circ \phi \circ \xi} \cap \mathcal{H} \cap \Xi_{\mathcal{H}} \big \vert} {\big \vert  \mathcal{B}_n^{\xi^{-1} \circ \phi \circ \xi} \big \vert},
\end{align}
where $ \Xi_{\mathcal{H}}$ is the subset of $\mathcal{H}$ the corresponding RoughP algorithm, $\mathcal{P}_{\mathcal{H}}$ (equivalent to the algorithm in Ref.~\cite[Eqn.~4]{farago2016roughly}), decides correctly (i.e. the complement of the set in Eqn.~\ref{Def:RoughPforLProperty}).

I begin this task by, using that $\xi: \Sigma^* \longrightarrow \Pi^*$ preserves the ordering on $\Sigma^*$ following from viewing $\Sigma^*$ as encoding natural numbers (in base $\vert \Sigma \vert$) as in Def.~\ref{def:OmegaSigma} and Def.~\ref{def:alphaPi} -- providing the ordering on $\Pi^*$ that follows from viewing $\Pi^*$ as encoding natural numbers (this time in base $\vert \Pi \vert$)~\footnote{and hence so does $\xi^{-1}: \Pi \longrightarrow \Sigma^*$} -- and how $\mathcal{H}$ was defined to make $\xi$ a preserving-P-isomorphism, I can use the first ``additional property" in Def.~\ref{def:adequete} to derive~\footnote{I note that the summation over multiple $\mathcal{B}_j^{\phi}$ is required because the set of words of $\Pi^*$ of a given length is larger than the set of words of $\Sigma^*$ of the same length (as $\vert \Sigma \vert < \vert \Pi \vert$).}:
\begin{align}
    \label{eqn:noHinBNwithximinusOnephixiAgain}
    \big \vert \mathcal{B}_n^{\xi^{-1} \circ \phi \circ \xi} \cap \mathcal{H} \cap \Xi_{\mathcal{H}} \big \vert
    &=
    \sum_{j = \chi_{low}(n)}^{\chi_{upp}(n)} \bigg( \alpha_j \big \vert \mathcal{B}_j^{\phi} \cap \mathcal{L} \cap \Xi_{\mathcal{L}} \big \vert \bigg),
\end{align}
where $\chi_{low}(n)$ and $\chi_{upp}(n)$ are integer-valued functions of $n$ that need not be defined \emph{yet}, and $\big\{ \alpha_j \in \mathbb{R}\big\}_{j = \chi_{low}(n)}^{\chi_{upp}(n)}$ satisfies:
\begin{align}
    \label{eqn:DefiningPropertyOfalphasSum}
    \sum_{j = \chi_{low}(n)}^{\chi_{upp}(n)} \bigg( \alpha_j \big \vert \mathcal{B}_j^{\phi} \big \vert \bigg) = \big \vert  \mathcal{B}_n^{\xi^{-1} \circ \phi \circ \xi} \big \vert,
\end{align}
such that $\alpha_j = 1$ unless $j = \chi_{low}(n)$ or $\chi_{upp}(n)$, in which case $\alpha_j \leq 1$.
Eqn.~\ref{eqn:noHinBNwithximinusOnephixiAgain} (along with the above mentioned -- in Eqn.~\ref{eqn:DefiningPropertyOfalphasSum} -- properties of $\{ \alpha_j \}^{\chi_{upp}(n)}_{j = {\chi_{low}(n)}}$) is proven in Lemma~\ref{lem:RelationTheorem}.

The first step to bounding the right-hand-side of Eqn.~\ref{eqn:noHinBNwithximinusOnephixiAgain} is to note that $\mathcal{L}$ being adequately-balanced, via the third additional property in Def.~\ref{def:adequete}, implies it is alt-NAEU. $\big( \xi \circ \mathcal{P}_{\mathcal{H}} \big)$ then provides the RoughP algorithm -- that decides $\mathcal{L}$, by using a preserving-P-isomorphism, $\xi$, to map input strings to $\Pi^*$ -- for $\mathcal{L}$ being alt-NAEU to imply that: 
\begin{align}
\label{eqn:UsingNotexpInProofOfNotExp}
\big \vert \mathcal{B}_j^{\phi} \cap \mathcal{L} \cap \Xi_{\mathcal{L}}\big \vert \geq
\dfrac{1}{\textit{Poly}(j) } \big \vert \mathcal{B}_j^{\phi}\big \vert,
\end{align}
for some polynomial, $\textit{Poly}: \mathbb{N}_0 \longrightarrow \mathbb{R}$, meeting the monotonicity requirement. Where $\Xi_{\mathcal{L}}$ is (as shown in Eqn.~\ref{Def:RoughPforLProperty}) both the subset of $\Sigma^*$ that $\mathcal{P}_{\mathcal{L}}$ does not return $\bot$ for and $\xi^{-1} \big( \Xi_{\mathcal{H}} \big)$. Therefore, using Eqn.~\ref{eqn:noHinBNwithximinusOnephixiAgain}:
\begin{align}
    \dfrac{\big \vert \mathcal{B}_n^{\xi^{-1} \circ \phi \circ \xi} \cap \mathcal{H} \cap \Xi_{\mathcal{H}} \big \vert} {\big \vert  \mathcal{B}_n^{\xi^{-1} \circ \phi \circ \xi} \big \vert}
    &=
    \dfrac{1} {\big \vert  \mathcal{B}_n^{\xi^{-1} \circ \phi \circ \xi} \big \vert}
    \sum_{j = \chi_{low}(n)}^{\chi_{upp}(n)} \bigg( \alpha_j \big \vert \mathcal{B}_j^{\phi} \cap \mathcal{L} \cap \Xi_{\mathcal{L}} \big \vert \bigg)
    \geq
    \dfrac{1} {\big \vert  \mathcal{B}_n^{\xi^{-1} \circ \phi \circ \xi} \big \vert}
    \sum_{j = \chi_{low}(n)}^{\chi_{upp}(n)} \bigg( \dfrac{\alpha_j}{\textit{Poly}(j) } \big \vert \mathcal{B}_j^{\phi}\big \vert\bigg) \nonumber\\
    \label{eqn:notanywhereProofOfHBeforeSub}
    &\geq
    \dfrac{1} {\big \vert  \mathcal{B}_n^{\xi^{-1} \circ \phi \circ \xi} \big \vert}
    \sum_{j = \chi_{low}(n)}^{\chi_{upp}(n)} \bigg( \dfrac{\alpha_j}{\textit{Poly}(\chi_{upp}(n)) } \big \vert \mathcal{B}_j^{\phi}\big \vert\bigg).
\end{align}
Eqn.~\ref{eqn:notanywhereProofOfHBeforeSub} then implies that, using Eqn.~\ref{eqn:DefiningPropertyOfalphasSum}:
\begin{align}
    \dfrac{\big \vert \mathcal{B}_n^{\xi^{-1} \circ \phi \circ \xi} \cap \mathcal{H} \cap \Xi_{\mathcal{H}} \big \vert} {\big \vert  \mathcal{B}_n^{\xi^{-1} \circ \phi \circ \xi} \big \vert}
    &\geq
    \dfrac{1}{\textit{Poly}(\chi_{upp}(n)) } \dfrac{1} {\big \vert  \mathcal{B}_n^{\xi^{-1} \circ \phi \circ \xi} \big \vert}
    \sum_{j = \chi_{low}(n)}^{\chi_{upp}(n)} \bigg(\alpha_j \big \vert \mathcal{B}_j^{\phi}\big \vert\bigg)
    =
    \dfrac{1}{\textit{Poly}(\chi_{upp}(n))} \dfrac{\big \vert  \mathcal{B}_n^{\xi^{-1} \circ \phi \circ \xi} \big \vert} {\big \vert  \mathcal{B}_n^{\xi^{-1} \circ \phi \circ \xi} \big \vert}
    =
    \dfrac{1}{\textit{Poly}(\chi_{upp}(n))}.
\end{align}
To progress any further, I first need to bound (from above) $\chi_{upp}(n)$. This is done in Lemma~\ref{lem:RelationTheorem} and it shows that $\chi_{upp}(n) = \mathcal{O} \big( n \big)$.

Therefore, $\textit{Poly}(\chi_{upp}(n))$ can be viewed as a polynomial in $n$, which I refer to as $\textit{Poly}': \mathbb{N}_0 \longrightarrow \mathbb{R}$. This then implies that there exists a polynomial, $\textit{Poly}': \mathbb{N}_0 \longrightarrow \mathbb{R}$, such that:
\begin{align}
    \dfrac{\big \vert \mathcal{B}_n^{\xi^{-1} \circ \phi \circ \xi} \cap \mathcal{H} \cap \Xi_{\mathcal{H}}\big \vert} {\big \vert  \mathcal{B}_n^{\xi^{-1} \circ \phi \circ \xi} \big \vert}
    \geq
    \dfrac{1}{\textit{Poly}'(n)}.
\end{align}
The second additional property in Def.~\ref{def:adequete} ensures that $\textit{Poly}'(n) \big( 1 / \sqrt{2} \big)^n$ is also monotonically decreasing, as required. This is shown in Lemma~\ref{lem:polyBoundedIfPoly}, which is presented immediately after Lemma~\ref{lem:RelationTheorem}, after the end of this proof.

As $\mathcal{L}^c$ is also adequately-balanced if $\mathcal{L}$ is (as the definition of being adequately-balanced -- in Def.~\ref{def:adequete} -- is unaffected by interchanging $\mathcal{L}$ and $\mathcal{L}^c$), the same argument as above implies that:
there exists a polynomial, $\textit{Poly}'': \mathbb{N}_0 \longrightarrow \mathbb{R}$, such that:
\begin{align}
    \dfrac{\big \vert \mathcal{B}_n^{\xi^{-1} \circ \phi \circ \xi} \cap \mathcal{H}^c \cap \Xi_{\mathcal{H}^c} \big \vert} {\big \vert  \mathcal{B}_n^{\xi^{-1} \circ \phi \circ \xi} \big \vert}
    \geq
    \dfrac{1}{\textit{Poly}''(n)},
\end{align}
and the second additional property in Def.~\ref{def:adequete} again ensures that $\textit{Poly}''(n) \big( 1 / \sqrt{2} \big)^n$ is also monotonically decreasing.
\newline

\noindent So $\mathcal{H}$ is not-anywhere-exponentially-unbalanced.
\newline

\noindent \underline{\textbf{Conclusion of the Proof}}\\
The language $\mathcal{H}$ has been shown to be paddable, not-anywhere-exponentially-unbalanced, and preserving-P-isomorphic to the language $\mathcal{L} \subseteq \Sigma^*$.

As the choice of $\mathcal{L}$ and $\Sigma$ was completely arbitrary, this shows that any adequately-balanced language over an odd-sized alphabet is preserving-P-isomorphic to a paddable not-anywhere-exponentially-unbalanced language over an even-sized alphabet. I.e. for any $\mathcal{L} \subseteq \Sigma^*$, a $\mathcal{H} \subseteq \Pi^*$, as defined above, exists. 
\end{proof}
\subsection{Proving Lemma~\ref{lem:PisoImpliesPhase}}
\begin{lemma}
\label{lem:RelationTheorem}
If $\xi, \phi, \mathcal{L}$, and $\mathcal{H}$ are as in the proof of Lemma~\ref{lem:PIsoExists}, $\forall n \in \mathbb{N}_0$, there exists $\chi_{low}(n), \chi_{upp}(n) \in \mathbb{N}_0$ such that:
    \begin{align}
    \big \vert \mathcal{B}_n^{\xi^{-1} \circ \phi \circ \xi} \cap \mathcal{H} \cap \Xi_{\mathcal{H}} \big \vert
    &=
    \sum_{j = \chi_{low}(n)}^{\chi_{upp}(n)} \bigg( \alpha_j \big \vert \mathcal{B}_j^{\phi} \cap \mathcal{L} \cap \Xi_{\mathcal{L}} \big \vert \bigg),
\end{align}
where $\big\{ \alpha_j \in \mathbb{R}^+ \big\}_{j = \chi_{low}(n)}^{\chi_{upp}(n)}$ satisfies
    $\sum_{j = \chi_{low}(n)}^{\chi_{upp}(n)} \bigg( \alpha_j \big \vert \mathcal{B}_j^{\phi} \big \vert \bigg) = \big \vert  \mathcal{B}_n^{\xi^{-1} \circ \phi \circ \xi} \big \vert$ with $\alpha_j = 1$ unless $j = \chi_{low}(n)$ or $\chi_{upp}(n)$, in which case $\alpha_j \leq 1$; and $ \Xi_{\mathcal{H}}$ is the subset of $\Pi^*$ the corresponding RoughP algorithm that decides $\mathcal{H}$, $\mathcal{P}_{\mathcal{H}}$, decides correctly (i.e. does not return $\bot$ for).

Furthermore, $\chi_{upp}(n)
    \leq
    2n \log_{\vert \Sigma \vert} \bigg( \vert \Sigma \vert + 1 \bigg) = \mathcal{O} \big( n \big)$.
\end{lemma}
\begin{proof}
    Let $\theta_{\Sigma}$ and $\theta_{\Pi}$ retain their meaning from the proof of Lemma~\ref{lem:BasicIsomorpLemma}, and let $x, y \in \Sigma^*$ such that:
    \begin{align}
        \label{eqn:firstPlusOneEqn}
        \theta_{\Sigma} \big( \phi (x) \big)
        =
        \theta_{\Sigma} \big( \phi (y) \big) + 1.
    \end{align}
    For any $y \in \Sigma^*$, there exists a $x \in \Sigma^*$ that satisfies Eqn.~\ref{eqn:firstPlusOneEqn}, and for any $x \in \Sigma^*$ such that $\theta_{\Sigma} \big( \phi \big( x \big) \big) \not = 0$, there exists a $y \in \Sigma^*$ that satisfies Eqn.~\ref{eqn:firstPlusOneEqn}.
    By diagram chasing in Fig.~\ref{fig:FinalProofDiagram}, $\phi \circ \theta_{\Sigma}$ may be re-expressed as:
    \begin{align}
        \phi \circ \theta_{\Sigma} 
        &=
        \xi \circ \big( \xi^{-1} \circ \phi \circ \xi \big) \circ \theta_{\Pi}
        =
         \phi \circ \xi \circ \theta_{\Pi}.
    \end{align}
    Therefore, letting $x_0, y_0 \in \Pi^*$ such that $x = \xi^{-1} \big( x_0 \big)$ and $y = \xi^{-1} \big( y_0 \big)$ (as $\xi$ is bijective, such $x_0$ and $y_0$ always exist), Eqn.~\ref{eqn:firstPlusOneEqn} implies that:
    \begin{align}
        \label{def:beforeToExplain}
        \theta_{\Pi} \big( \xi \big( \phi \big( \xi^{-1} \big( x_0 \big) \big) \big) \big)
        =
        \theta_{\Pi} \big( \xi \big( \phi \big( \xi^{-1} \big( y_0 \big) \big) \big) \big) + 1.
    \end{align}
    Eqn.~\ref{def:beforeToExplain} may be expressed alternatively as:
    \begin{align}
        \label{eqn:subsetContainmentofB}
        \big( \xi^{-1} \circ \phi \circ \xi \circ \theta_{\Pi} \big) \big( x_0\big)
        &=
        \big( \xi^{-1} \circ \phi \circ \xi \circ \theta_{\Pi} \big) \big( y_0\big) + 1.
    \end{align}
    Eqn.~\ref{eqn:subsetContainmentofB} following from Eqn.~\ref{eqn:firstPlusOneEqn} is interpreted, informally, as, for any $\mathcal{S} \subseteq \Sigma^*$ such that $\phi$ maps them to the subset of $\Sigma^*$ encoding exactly the integers in the a certain range (e.g. $\mathcal{B}^{\phi}_n$, for any $n \in \mathbb{N}_0$), $\xi \big( \mathcal{S} \big) \subseteq \Pi^*$ is mapped to the subset of $\Pi^*$ encoding exactly the integers in the same range by $\xi^{-1} \circ \phi \circ \xi$   
    (e.g. a contiguous subset of $\mathcal{B}^{\xi^{-1} \circ \phi \circ \xi}_n$, for a particular $n \in \mathbb{N}_0$)~\footnote{Also note that the argument above goes through if reversed.}.
    The most important instances of these relations are the mentioned examples.
    \newline
    
    Using the above argument, as:
    \begin{enumerate}
        \item the elements of $\mathcal{B}^{\phi}_n$ after $\phi$ acts on them are contiguous~\footnote{By which I mean there are not elements not in that set between elements of that set.} -- according to the ordering that follows from $\theta_{\Sigma}$ -- and are sandwiched between the elements of $\mathcal{B}^{\phi}_{n-1}$ (on one side) and $\mathcal{B}^{\phi}_{n+1}$ (on the other);
        \item the elements of $\mathcal{B}^{\xi^{-1} \circ \phi \circ \xi}_n$ after $\xi^{-1} \circ \phi \circ \xi$ acts on them are also contiguous, this time according to the ordering that follows from $\theta_{\Pi}$;
        \item $\xi$ and $\xi^{-1}$ preserve the orderings of $\Sigma^*$  according to $\theta_{\Sigma}$ and  the orderings of $\Pi^*$ according to $\theta_{\Pi}$, respectively~\footnote{As they are designed to encode exactly the same numbers before and after $\xi$ or $\xi^{-1}$ act on them, so $\xi$ and $\xi^{-1}$ preserve the ordering inherited from the encoding of the numbers.};
        \end{enumerate}
        there exists $\chi_{low}(n), \chi_{upp}(n) \in \mathbb{N}_0$ such that:
    \begin{align}
        \label{eqn:FirstRelationshipEquation}
        \xi^{-1} \bigg( \mathcal{B}^{\xi^{-1} \circ \phi \circ \xi}_n \bigg) \subseteq \bigcup_{k = \chi_{low}(n)}^{\chi_{upp}(n) } \bigg( \mathcal{B}^{\phi}_k \bigg)
        \hspace{0.6cm}
        \textit{ and }
        \hspace{0.6cm}
        \xi \bigg( \bigcup_{k = \chi_{low}(n)}^{\chi_{upp}(n) } \bigg( \mathcal{B}^{\phi}_k \bigg) \bigg) \subseteq \mathcal{B}^{\xi^{-1} \circ \phi \circ \xi}_n.
    \end{align}
    Switching focus momentarily, to refine the relationship in Eqn.~\ref{eqn:FirstRelationshipEquation}, let $e_0 \in \mathcal{B}^{\phi}_0$. This is equivalent, given $e_0 \in \Sigma^*$, to there not existing a $y_0 \in \Sigma^*$ such that:
    \begin{align}
        \label{eqn:qwerty}
        \theta_{\Sigma} \big( \phi \big( y_0 \big) \big) = \theta_{\Sigma} \big( \phi \big( e_0 \big) \big) - 1.
    \end{align}
    By diagram chasing on Fig.~\ref{fig:FinalProofDiagram}, Eqn.~\ref{eqn:qwerty} can be seen to be equivalent to there not existing a $y_0 \in \Sigma^*$ such that:
    \begin{align}
        \label{eqn:qwertyTwo}
        \theta_{\Pi} \big( \xi \big( \phi \big( y_0 \big) \big) \big) = \theta_{\Pi} \big( \xi \big( \phi \big( e_0 \big) \big) \big) - 1.
    \end{align}
    Such a $y_0 \in \Sigma^*$ satisfying Eqn.~\ref{eqn:qwertyTwo} would be implied to exist by there existing $y_0', e_0' \in \Pi^*$ such that:
    $y_0 = \xi^{-1}\big( y_0' \big)$, $e_0 = \xi^{-1}\big( e'_0\big)$, and:
    \begin{align}
        \big( \xi^{-1} \circ \phi \circ \xi \circ \theta_{\Pi} \big)\big( y_0' \big)
        =
        \big( \xi^{-1} \circ \phi \circ \xi \circ \theta_{\Pi} \big)\big( e_0' \big) - 1.
    \end{align}
    So no such $y_0', e_0' \in \Pi^*$ can exist  (assuming $e_0 \in \mathcal{B}^{\phi}_0$). This is only possible if $\big( \xi^{-1} \circ \phi \circ \xi \circ \theta_{\Pi} \big)\big( e_0' \big)$ is the empty word in $\Pi^*$. Therefore, $e_0' \in \mathcal{B}_0^{\xi^{-1} \circ \phi \circ \xi}$.

    In fact, as $\big \vert \mathcal{B}_0^{\xi^{-1} \circ \phi \circ \xi} \big \vert = \big \vert \mathcal{B}_0^{ \phi} \big \vert = 1$, and $e_0 = \xi^{-1} \big( e'_0 \big)$, the non-existence of a $y_0', e_0' \in \Pi^*$ as required implies that:
    \begin{align}
        \label{eqn:Aligning}
        \xi^{-1} \bigg( \mathcal{B}_0^{\xi^{-1} \circ \phi \circ \xi} \bigg)
        =
        \mathcal{B}_0^{\phi}.
    \end{align}
    As Eqn.~\ref{eqn:FirstRelationshipEquation} holds for all values of $n \in \mathbb{N}_0$, Eqn.~\ref{eqn:Aligning} serves as a base case which allows me to inductively conclude that the relationship between $\big \{ \mathcal{B}_k^{\xi^{-1} \circ \phi \circ \xi} \big \}_{k = 0}^{\infty}$ and $\big \{ \mathcal{B}_k^{\phi} \big \}_{k = 0}^{\infty}$ is almost entirely based on the respective sizes of the $\mathcal{B}_k^{\phi}$ and the $\mathcal{B}_k^{\xi^{-1} \circ \phi \circ \xi}$. As the ordering of the domain and range of $\xi$ is preserved by it, due to the numbers encoded being preserved by $\xi$ (and the number obviously implying an ordering); as encodings of neighboring numbers have either the same length or lengths that differ by one (which implies Eqn.~\ref{eqn:FirstRelationshipEquation}), Eqn.~\ref{eqn:Aligning} provides a base and then $\Sigma^*$ and $\Pi^*$ can each be placed, in -- an already discussed -- order, into their respective hierarchies of $\mathcal{B}_k^{\phi}$ and the $\mathcal{B}_k^{\xi^{-1} \circ \phi \circ \xi}$, filling each set in the respective hierarchies up progressively. So, deciding which $\mathcal{B}_k^{\xi^{-1} \circ \phi \circ \xi}$ an element of any $\mathcal{B}_k^{\phi}$ is in becomes a matter of simply expressing the encoded number in base $\vert \Pi \vert$ and counting its length (which gives you the required $k \in \mathbb{N_0}$).
    \newline

    In a brief interlude that will be required imminently, consider how Ref.~\cite[Theorem 1]{farago2016roughly} implies a RoughP algorithm exists to decide $\mathcal{H}$ (as it is paddable and hence in RoughP), $\mathcal{P}_{\mathcal{H}}$, which, via $\xi$ being a preserving-P-isomorphism, implies a RoughP algorithm to decide $\mathcal{L}$, $\mathcal{P}_{\mathcal{L}}$, defined by: $\forall x \in \Sigma^*$,
\begin{align}
    \label{Def:RoughPforL}
    \mathcal{P}_{\mathcal{L}}  \big( x \big)  
    &=
    \mathcal{P}_{\mathcal{H}} \big( \xi (x) \big)
    =
    (\xi \circ  \mathcal{P}_{\mathcal{H}}) (x). 
\end{align}
Therefore, 
\begin{align}
    \label{Def:RoughPforLProperty}
    \xi \big [ \big \{ x \in \Sigma^* \text{ } \big \vert \text{ }  \mathcal{P}_{\mathcal{L}}\big( x \big) = \perp \big \} \big]
    &=
    \big \{ \xi \big [x \big] \in \Pi^* \text{ } \big \vert \text{ }  \mathcal{P}_{\mathcal{H}} \big( \xi [x] \big) = \perp \big \}
    =
    \big \{ y \in \Pi^* \text{ } \big \vert \text{ } 
    \mathcal{P}_{\mathcal{H}} \big( \xi \big [ \xi^{-1}[y] \big] \big) = \perp \big \}
     =
    \big \{ y \in \Pi^* \text{ } \big \vert \text{ } 
    \mathcal{P}_{\mathcal{H}} \big( y \big) = \perp \big \}.
\end{align}
That is, $\xi: \Sigma^* \longrightarrow \Pi^*$ maps the elements of $\Sigma^*$ that $\mathcal{P}_{\mathcal{L}}$ cannot decide to elements of $\Pi^*$ that $\mathcal{P}_{\mathcal{H}}$ cannot decide while preserving membership of the respective languages. Or, expressed alternatively: $\Xi_{\mathcal{L}} = \xi^{-1} \big( \Xi_{\mathcal{H}} \big)$ and $\Xi_{\mathcal{H}} = \xi \big( \Xi_{\mathcal{L}} \big)$. 

    Returning to the main thread of this proof:
    Eqn.~\ref{eqn:FirstRelationshipEquation}, Eqn.~\ref{eqn:Aligning}, Eqn.~\ref{Def:RoughPforL}, Eqn.~\ref{Def:RoughPforLProperty}, and the first additional requirement of Def.~\ref{def:adequete} imply that (as each of the sets that intersect to form $\mathcal{B}_n^{\xi^{-1} \circ \phi \circ \xi} \cap \mathcal{H} \cap \Xi_{\mathcal{H}}$ are mapped to their equivalent for $\mathcal{L}$ by $\xi^{-1}: \Pi^* \longrightarrow \Sigma^*$):
    \begin{align}
    \label{eqn:noHinBNwithximinusOnephixiAgainTwo}
    \big \vert \mathcal{B}_n^{\xi^{-1} \circ \phi \circ \xi} \cap \mathcal{H} \cap \Xi_{\mathcal{H}} \big \vert
    &=
    \sum_{j = \chi_{low}(n)}^{\chi_{upp}(n)} \bigg( \alpha_j \big \vert \mathcal{B}_j^{\phi} \cap \mathcal{L} \cap \Xi_{\mathcal{L}} \big \vert \bigg),
    \end{align}
    where $\big\{ \alpha_j \in \mathbb{R}^+ \big\}_{j = \chi_{low}(n)}^{\chi_{upp}(n)}$ satisfies:
\begin{align}
    \sum_{j = \chi_{low}(n)}^{\chi_{upp}(n)} \bigg( \alpha_j \big \vert \mathcal{B}_j^{\phi} \big \vert \bigg) = \big \vert  \mathcal{B}_n^{\xi^{-1} \circ \phi \circ \xi} \big \vert,
\end{align}
with the additional condition that: $\alpha_j = 1$ unless $j = \chi_{low}(n)$ or $\chi_{upp}(n)$, in which case $\alpha_j \leq 1$.
\newline

I now turn to consider the last claim of the lemma statement, concerning $\chi_{upp}(n)$.
\newline

$\chi_{upp}(n)$, as in Eqn.~\ref{eqn:noHinBNwithximinusOnephixiAgainTwo}, can be calculated using the relationship captured by Eqn.~\ref{eqn:subsetContainmentofB} and Eqn.~\ref{eqn:Aligning}, i.e. the ``filling up" relationship.
As Eqn.~\ref{eqn:noHinBNwithximinusOnephixiAgainTwo} was almost entirely determined by how $\xi: \Sigma^* \longrightarrow \Pi^*$ was constructed and then how $\mathcal{H}$ was defined to make $\xi$ a preserving-P-isomorphism, this is akin to an analysis of the length of a word from an alphabet of size $\vert \Sigma \vert \in \mathbb{N}$ needed to encode the the largest number representable by a length $n \in \mathbb{N}_0$ word from an alphabet with size $\vert \Sigma \vert + 1 \in \mathbb{N}$:  
\begin{align}
    \chi_{upp}(n)
    &=
    \bigg \lfloor \log_{\vert \Sigma \vert} \bigg( (n - 1) \sum^{n-1}_{j = 0} \bigg[ \big( \vert \Sigma \vert + 1 \big)^j \bigg] \bigg) \bigg \rfloor + 1
    =
    \bigg \lfloor \dfrac{\ln{\big( n - 1 \big)} + \ln{\big(  \big( \vert \Sigma \vert + 1 \big)^n - 1 \big)}- \ln{ \big( \vert \Sigma \vert \big)}}{\ln{\big( \vert \Sigma \vert \big)}} \bigg \rfloor + 1\\
    &=
    \bigg \lfloor \dfrac{\ln{\big( n - 1 \big)} + \ln{\big(  \big( \vert \Sigma \vert + 1 \big)^n - 1 \big)}}{\ln{\big( \vert \Sigma \vert \big)}} - 1 \bigg \rfloor + 1
    \leq
    \dfrac{\ln{\big( n - 1 \big)} + \ln{\big(  \big( \vert \Sigma \vert + 1 \big)^n \big)}}{\ln{\big( \vert \Sigma \vert \big)}}
    \leq
    2n \dfrac{\ln{\big( \vert \Sigma \vert + 1 \big)}}{\ln{\big( \vert \Sigma \vert \big)}}
    =
    2n \log_{\vert \Sigma \vert} \bigg( \vert \Sigma \vert + 1 \bigg).
\end{align}
$\chi_{low}(n)$ can be calculated similarly but this is not required. However, $\chi_{low}(n) = \mathcal{O} \big( n \big)$ and $\chi_{low}(n) = \Omega \big( n \big)$.
\end{proof}

\begin{figure}[h!]
    \centering

    \begin{tikzcd}[scale cd=1.9, row sep=huge,column sep=huge]
\mathcal{L}^c \arrow[rr, "\xi", bend left] \arrow[dd, "\phi"']                 &              & \mathcal{H}^c \arrow[ll, "\xi^{-1}", bend left] \arrow[dd, "\xi^{-1} \circ \phi \circ \xi"]  \\
                                                                             &              &                                                                                            \\
H_{\mathcal{L}}^c \arrow[rr, "\xi", bend left] \arrow[rdd, "\theta_{\Sigma}"'] &              & H_{\mathcal{H}}^c \arrow[ll, "\xi^{-1}", bend left] \arrow[ldd, "\theta_{\Pi}"]              \\
                                                                             &              &                                                                                            \\
                                                                             & \mathbb{N}_0 &                                                                                            \\
                                                                             &              &                                                                                            \\
H_{\mathcal{L}} \arrow[rr, "\xi", bend left] \arrow[ruu, "\theta_{\Sigma}"]  &              & H_{\mathcal{H}} \arrow[ll, "\xi^{-1}", bend left] \arrow[luu, "\theta_{\Pi}"']             \\
                                                                             &              &                                                                                            \\
\mathcal{L} \arrow[rr, "\xi", bend left] \arrow[uu, "\phi"]                  &              & \mathcal{H} \arrow[ll, "\xi^{-1}", bend left] \arrow[uu, "\xi^{-1} \circ \phi \circ \xi"']
\end{tikzcd}
\caption{\label{fig:FinalProofDiagram}
Diagram used in the proof of Lemma~\ref{lem:RelationTheorem}.
Note that $H_{\mathcal{L}}$ is defined from $H_{\mathcal{H}}$ (which is defined from $\mathcal{H}$ by the method in Ref.~\cite{farago2016roughly}) by the commutativity of this diagram.}
\end{figure}
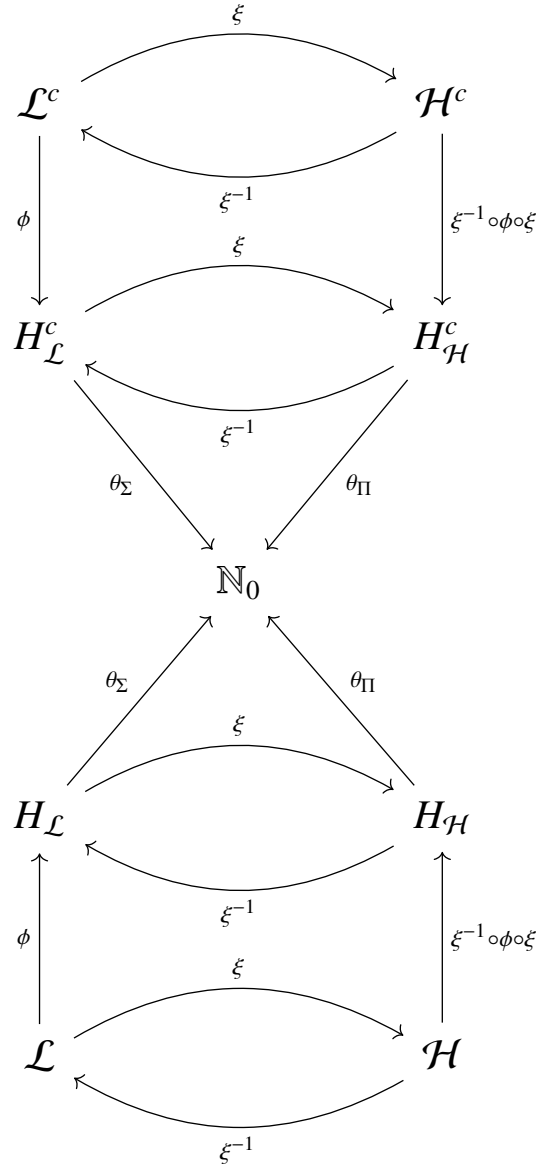

\begin{lemma}
    \label{lem:polyBoundedIfPoly}
    If $\textit{Poly(n)} \big( 1 / \sqrt{2} \big)^n$ is monotonically decreasing and, $\forall n \in \mathbb{N}_0$,
    \begin{align}
        \dfrac{\partial}{ \partial x} \bigg( \textit{Poly}(x) \bigg) \bigg \vert_{ x = 2n \log_{\vert \Sigma \vert} \big( \vert \Sigma \vert + 1 \big)}
        \leq
        \dfrac{\big \vert \ln{\big( 1 / \sqrt{2} \big)} \big \vert}{2 \log_{\vert \Sigma \vert} \big( \vert \Sigma \vert + 1 \big)} \textit{Poly} \big( 2n \log_{\vert \Sigma \vert} \big( \vert \Sigma \vert + 1 \big) \big),
    \end{align}
    then
    $\textit{Poly} \big( 2n \log_{\vert \Sigma \vert} \big( \vert \Sigma \vert + 1 \big) \big) \big( 1 / \sqrt{2} \big)^n$ is monotonically decreasing.
\end{lemma}
\begin{proof}
    Consider $\textit{Poly} \big( 2n \log_{\vert \Sigma \vert} \big( \vert \Sigma \vert + 1 \big) \big) \big( 1 / \sqrt{2} \big)^n$ and take the derivative with respect to $n$:
    \begin{align}
        \dfrac{\partial}{ \partial n} \bigg(\textit{Poly} \big( 2n \log_{\vert \Sigma \vert} \big( \vert \Sigma \vert + 1 \big) \big) \big( 1 / \sqrt{2} \big)^n \bigg)
        =&
        2 \log_{\vert \Sigma \vert} \big( \vert \Sigma \vert + 1 \big) \big( 1 / \sqrt{2} \big)^n \dfrac{\partial}{ \partial x} \bigg( \textit{Poly}(x) \bigg) \bigg \vert_{ x = 2n \log_{\vert \Sigma \vert} \big( \vert \Sigma \vert + 1 \big)} \nonumber\\
        &+
        \label{eqn:FirstForPropertyTwo}
        \ln{\big( 1 / \sqrt{2} \big)}\textit{Poly} \big( 2n \log_{\vert \Sigma \vert} \big( \vert \Sigma \vert + 1 \big) \big) \big( 1 / \sqrt{2} \big)^n.
    \end{align}
    For $\textit{Poly} \big( 2n \log_{\vert \Sigma \vert} \big( \vert \Sigma \vert + 1 \big) \big) \big( 1 / \sqrt{2} \big)^n$ to be monotonically decreasing, Eqn.~\ref{eqn:FirstForPropertyTwo} must be less than or equal to zero for all $n \in \mathbb{N}$.
    Rearranging the required inequality, $\textit{Poly} \big( 2n \log_{\vert \Sigma \vert} \big( \vert \Sigma \vert + 1 \big) \big) \big( 1 / \sqrt{2} \big)^n$ being monotonically decreasing is equivalent to:
    \begin{align}
        \label{eqn:FirstForPropertyTwoNumberTwo}
        2 \log_{\vert \Sigma \vert} \big( \vert \Sigma \vert + 1 \big) \big( 1 / \sqrt{2} \big)^n \dfrac{\partial}{ \partial x} \bigg( \textit{Poly}(x) \bigg) \bigg \vert_{ x = 2n \log_{\vert \Sigma \vert} \big( \vert \Sigma \vert + 1 \big) \big)}
        &\leq
        \big \vert \ln{\big( 1 / \sqrt{2} \big)} \big \vert\textit{Poly} \big( 2n \log_{\vert \Sigma \vert} \big( \vert \Sigma \vert + 1 \big) \big) \big( 1 / \sqrt{2} \big)^n\\
        \iff
        \dfrac{\partial}{ \partial x} \bigg( \textit{Poly}(x) \bigg) \bigg \vert_{ x = 2n \log_{\vert \Sigma \vert} \big( \vert \Sigma \vert + 1 \big) \big)}
        &\leq
        \dfrac{\big \vert \ln{\big( 1 / \sqrt{2} \big)} \big \vert}{2 \log_{\vert \Sigma \vert} \big( \vert \Sigma \vert + 1 \big)} \textit{Poly} \big( 2n \log_{\vert \Sigma \vert} \big( \vert \Sigma \vert + 1 \big) \big).
    \end{align}
    Eqn.~\ref{eqn:FirstForPropertyTwoNumberTwo}, by assumption, holds. Therefore, $\textit{Poly} \big( 2n \log_{\vert \Sigma \vert} \big( \vert \Sigma \vert + 1 \big) \big) \big( 1 / \sqrt{2} \big)^n$ is monotonically decreasing.
\end{proof}

\begin{lemma}
    \label{lem:PisoImpliesPhase}
	Any language that is preserving-P-isomorphic to a language that exhibits a phase transition also exhibits a phase transition.
\end{lemma}
\begin{proof}
Let $\Sigma$ and $\Pi$ be alphabets.
	Additionally, let $\mathcal{L} \subset \Sigma^*$ be a language and $\xi: \Sigma^* \longrightarrow \Pi^*$ be a preserving-P-isomorphism from $\mathcal{L}$ to $\mathcal{H}$, where $\mathcal{H} \subset \Pi^*$ is a language that exhibits a phase transition.
    
	Suppose $\gamma: \Pi^* \longrightarrow \mathbb{R}$ is the parameter that induces the phase transition in $\mathcal{H}$. I then define another parameter, $\gamma': \Sigma^* \longrightarrow \mathbb{R}$, intended to induce a phase transition in $\mathcal{L}$ by: $\forall x \in \Sigma^*$,
	\begin{align}
		\gamma'(x)
		=
		\big( \xi \circ \gamma \big)(x).
	\end{align}
	To show that $\gamma': \Sigma^* \longrightarrow \mathbb{R}$ induces a phase transition in $\mathcal{L}$, I consider each condition (abbreviated $\bold{Cond.}$ below) of phase transitions (as defined in Def.~\ref{PhaseDef}), in turn, and show they are met:\\
    
\begin{enumerate}
	\item[\underline{\textbf{Cond. 1:}}] Define $\Tilde{S}^{\gamma}_n = \big \{ y \in \Pi^* \text{ } \vert \text{ } \gamma(y) = n \big \} \subseteq \Pi^*$ (where the tilde is used to denote that $\Tilde{S}^{\gamma}_n$ is a subset of $\Pi^*$) and consider:
	\begin{align}
		S^{\gamma'}_n &= S^{\xi \circ \gamma}_n = \big \{ x \in \Sigma^* \text{ } \vert \text{ } ( \xi \circ \gamma )(x) = n \big \} \nonumber \\
        \label{eqn:equivOfSgammaStuff}
		\Rightarrow
		\xi \big( S^{\gamma'}_n \big) 
		&=
		 \big \{ \xi(x) \in \Pi^* \text{ } \vert \text{ } x \in \Sigma^* \text{ and } ( \xi \circ \gamma )(x) = n \big \}
		=
		 \big \{ y \in \Pi^* \text{ } \vert \text{ } ( \xi \circ \gamma )(\xi^{-1}(y)) = n \big \}
		 =
		 \big \{ y \in \Pi^* \text{ } \vert \text{ } \gamma(y) = n \big \}
		 =
		 \Tilde{S}^{\gamma}_n.
	\end{align}
	Switching focus, let $\mathcal{A}_{\mathcal{H}} \big[ \Tilde{S} \big]$ be the fraction of $\Tilde{S} \subseteq \Pi^*$ that is in $\mathcal{H}$; known as the acceptance fraction of $\mathcal{H}$ in $\Tilde{S}$. Similarly, let $\mathcal{A}_{\mathcal{L}} \big[ S \big]$ be the fraction of $S \subseteq \Sigma^*$ that is in $\mathcal{L}$. As, by assumption, $\mathcal{H}$ exhibits a phase transition:
	\begin{align}
        \label{eqn:limitingOfALp}
		\lim_{n \longrightarrow \infty} \bigg( \mathcal{A}_{\mathcal{H}} \big[ \Tilde{S}^{\gamma}_n \big] \bigg) 
		=
		1.
	\end{align}
	As $\xi$ is a preserving-P-isomorphism, $\xi$ is bijective and, $\forall x \in \Sigma^*$, $x \in \mathcal{L} \iff \xi(x) \in \mathcal{H}$. Therefore, also using that $\xi \big( S^{\gamma'}_n \big) = \Tilde{S}^{\gamma}_n$ (as shown in Eqn.~\ref{eqn:equivOfSgammaStuff}), the equation for $\mathcal{A}_{\mathcal{L}} \big[ S^{\gamma'}_n \big]$ (Eqn.~\ref{eqn:indef:acceptingFraction}) in Def.~\ref{def:acceptingFraction} implies that: $\forall n \in \mathbb{N}_0$,
	\begin{align}
        \label{eqn:equatingAs}
		\mathcal{A}_{\mathcal{L}} \big[ S^{\gamma'}_n \big]
		=
        \dfrac{\big \vert \mathcal{L} \cap S^{\gamma'}_n \big \vert}{\big \vert S^{\gamma'}_n \big \vert}
        =
        \dfrac{\big \vert \xi \big( \mathcal{L} \cap S^{\gamma'}_n \big) \big \vert}{\big \vert \xi \big( S^{\gamma'}_n \big) \big \vert}
        =
        \dfrac{\big \vert \xi \big( \mathcal{L} \big) \cap \xi \big( S^{\gamma'}_n \big) \big \vert}{\big \vert \xi \big( S^{\gamma'}_n \big) \big \vert}
        =
        \dfrac{\big \vert \mathcal{H} \cap \Tilde{S}^{\gamma}_n \big \vert}{\big \vert \Tilde{S}^{\gamma}_n \big \vert}
		=
		\mathcal{A}_{\mathcal{H}} \big[ \Tilde{S}^{\gamma}_n \big].
	\end{align}
	Hence, combining Eqn.~\ref{eqn:limitingOfALp} and Eqn.~\ref{eqn:equatingAs} gives the required $\mathbf{Cond~1}$ for $\mathcal{L}$:
	\begin{align}
		\lim_{n \longrightarrow \infty} \bigg( \mathcal{A}_{\mathcal{L}} \big[ S^{\gamma'}_n \big] \bigg)
		=
		\lim_{n \longrightarrow \infty} \bigg( \mathcal{A}_{\mathcal{H}} \big[ \Tilde{S}^{\gamma}_n \big] \bigg) 
		=
		1.
	\end{align}		  	 
	
	\item[\underline{\textbf{Cond. 2:}}] Using similar reasoning to above, as, by assumption, $\mathcal{H}$ exhibits a phase transition and $\xi$ is a preserving-P-isomorphism:
	\begin{align}
		\lim_{n \longrightarrow -\infty} \bigg( \mathcal{A}_{\mathcal{L}} \big[ S^{\gamma'}_n \big] \bigg)
		=
		\lim_{n \longrightarrow -\infty} \bigg( \mathcal{A}_{\mathcal{H}} \big[ \Tilde{S}^{\gamma}_n \big] \bigg) 
		=
		0.
	\end{align}
	
	\item[\underline{\textbf{Cond. 3:}}] As $\xi: \Sigma^* \longrightarrow \Pi^*$ is a preserving-P-isomorphism, it is a bijection, therefore: $\forall S \subseteq \Sigma^*$,
	\begin{align}
		\big \vert \xi (S) \big \vert 
		=
		\big \vert S \big \vert.
	\end{align}
	The next step first requires a brief diversion to define a useful subset, in Def.~\ref{def:SinRange}.
    \begin{definition}
        \label{def:SinRange}
    For a fixed $\delta \in \mathbb{R}$, $\forall A \in \mathbb{R}$, define $\underline{\Tilde{S}^{\gamma}_{A, A+\delta}}  \subset \Pi^*$ by:
	\begin{align}
		\Tilde{S}^{\gamma}_{A, A+\delta}
		=
		\big \{ x \in \Pi^* \text{ }\vert \text{ } \gamma(x) \in [A, A + \delta] \big \},
	\end{align}
    and similarly define $S^{\gamma'}_{A, A + \delta} = \big \{ y \in \Sigma^* \text{ } \vert \text{ } \gamma (\xi(y)) \in [A, A + \delta ] \big \}$.
        \end{definition}
        Resuming the main thread of this proof, and using Def.~\ref{def:SinRange}:
	\begin{align}
        \label{eqn:xi-1SThing}
		 \xi^{-1} \big( \Tilde{S}^{\gamma}_{A, A + \delta} \big)
		 &=
		 \big \{ \xi^{-1}(x) \in \Sigma^* \text{ } \vert \text{ } x \in \Pi^* \text{ and } \gamma(x) \in [A, A + 			\delta ] \big \}
		 =
		 \big \{ y \in \Sigma^* \text{ } \vert \text{ } \gamma (\xi(y)) \in [A, A + \delta ] \big \}
		 =
		 S^{\gamma'}_{A, A + \delta}
         \subseteq
         \Sigma^*.
	\end{align}
	As $\mathcal{H}$, by assumption, exhibits a phase transition, $\vert \Tilde{S}^{\gamma}_{A, A+\delta} \vert $ decays sufficiently quickly (i.e. exponentially) as $A$ approaches its own, corresponding, threshold to meet the third required condition of phase transitions. Therefore, using that $\xi^{-1}: \Pi^* \longrightarrow \Sigma^*$ is a P-isomorphism (as shown in Lemma~\ref{Lem:inverseIsPreserve}, in Appendix~\ref{sec:PIsoDetailsAppendix}), and therefore a bijection, Eqn.~\ref{eqn:xi-1SThing} implies:
	\begin{align}
		\big \vert \Tilde{S}^{\gamma}_{A, A+\delta} \big \vert
		=
		\big \vert \xi^{-1} \big( \Tilde{S}^{\gamma}_{A, A+\delta} \big) \big \vert
		=
		\big \vert S^{\gamma'}_{A, A + \delta} \big \vert.
	\end{align}
	I then conclude that $\big \vert S^{\gamma'}_{A, A + \delta} \big \vert$ must also decay sufficiently quickly, as $A$ approaches the threshold, and so meets the third required condition of phase transitions in Def.~\ref{PhaseDef}. 
\end{enumerate}
	I have therefore shown for each requirement of phase transitions (in Def.~\ref{PhaseDef}) that if $\mathcal{H}$ meets it then $\mathcal{L}$ also meets it by virtue of the preserving-P-isomorphism, $\xi: \Sigma^* \longrightarrow \Pi^*$, between them. So, if $\mathcal{H} \subseteq \Pi^*$ has a phase transition, so does $\mathcal{L} \subseteq \Sigma^*$.
\end{proof}

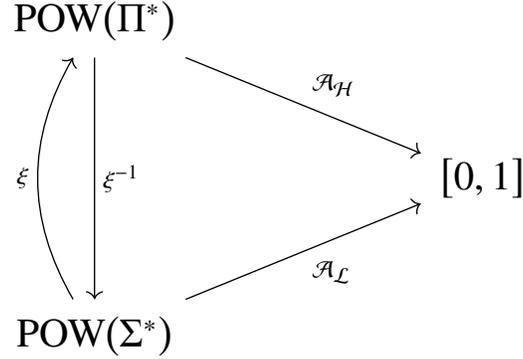
\begin{figure}[h!]
    \centering

    \begin{tikzcd}[scale cd=1.9, row sep=huge,column sep=huge]
\text{POW} \big( \Pi^* \big) \arrow[dd, "\xi^{-1}"] \arrow[rrd, "\mathcal{A}_{\mathcal{H}}"]           &  &                     \\
                                                                                                       &  & {\big [ 0, 1 \big]} \\
\text{POW} \big( \Sigma^* \big) \arrow[uu, "\xi", bend left] \arrow[rru, "\mathcal{A}_{\mathcal{L}}"'] &  &                    
\end{tikzcd}
    \caption{Diagram of the situation in Lemma~\ref{lem:PisoImpliesPhase}. Note that $\text{POW}$ denotes the powerset of its argument and that this diagram is commutative.}
    \label{fig:Lemma5Diagram}
\end{figure}

\section{The Incompatibility of Padding and Sparsity}
\label{app:PaddingSparsity}
The purpose of this appendix is to examine the relationship between sparsity and paddability, showing that they are incompatible and hence arguing -- informally -- that being paddable and not-anywhere-exponentially-unbalanced are commonly found together.

Before beginning this appendix, I first must define a particular subset of $\Sigma^*$, in Def.~\ref{def:SigmaN}.
\begin{definition}
    \label{def:SigmaN}
    $\forall n \in \mathbb{N}_0$, define $\Sigma^n \subseteq \Sigma^*$ by: $\Sigma^n
        =
        \big \{ x \in \Sigma^* \text{ } \vert \text{ } \vert x \vert = n \big \}$.
    Similarly, I define, $\Sigma^{\leq n} \subseteq \Sigma^*$ by:
    \begin{align}
        \Sigma^{\leq n}
        &=
        \bigcup_{j = 0}^n \bigg( \Sigma^j \bigg). 
    \end{align}
    
\end{definition}
Then I define the density of a language (in Def.~\ref{def:density}), which is the basis for sparsity (defined in Def.~\ref{def:sparseLanguages}).
\begin{definition}
    \label{def:density}
    For any alphabet, $\Sigma$, and any language, $\mathcal{L} \subseteq \Sigma^*$, the density of $\mathcal{L}$ in $ \Sigma^n$, $\textit{Dens}_{\mathcal{L}}: \mathbb{N}_0 \longrightarrow \mathbb{N}_0$, is defined by: $\forall n \in \mathbb{N}_0$,
    \begin{align}
        \textit{Dens}_{\mathcal{L}} \big( n \big) = \big \vert \mathcal{L} \cap \Sigma^{\leq n} \big \vert.
    \end{align}
\end{definition}
\begin{definition}
    \label{def:sparseLanguages}
    A language, $\mathcal{L} \subseteq \Sigma^*$, is sparse if and only if there exists a polynomial, $\textit{Poly}_{s}: \mathbb{N} \longrightarrow \mathbb{N}$, such that: $\forall n \in \mathbb{N}$,
    \begin{align}
        \textit{Dens}_{\mathcal{L}} \big( n \big) 
        \leq
        \textit{Poly}_{s}(n).
    \end{align}
\end{definition}
With the foundational concepts defined, I move to the crux of Appendix~\ref{app:PaddingSparsity}: a prohibition on languages being both sparse and paddable.
\begin{lemma}
    \label{lem:padSparseisNO}
    No paddable language (over an alphabet of size at least two) is sparse.
\end{lemma}
\begin{proof}(Based on Ref.~\cite{sparseAndPaddaple})\\
    Assume, for the sake of contradiction, $\mathcal{L} \subseteq \Sigma^*$ is a language (where $\vert \Sigma \vert \geq 2$) that is sparse (with the polynomial $\textit{Poly}_{s}: \mathbb{N}_0 \longrightarrow \mathbb{R}$ bounding the density of $\mathcal{L}$) and paddable (with padding function, $\textit{Pad}_{\mathcal{L}}: \Sigma^* \times \Sigma^* \longrightarrow \Sigma^*$).

    As, by definition, $\textit{Pad}_{\mathcal{L}}: \Sigma^* \times \Sigma^* \longrightarrow \Sigma^*$ runs in polynomial (in the size of the combined inputs) time, there exists a polynomial, $\textit{Poly}_{p}: \mathbb{N} \longrightarrow \mathbb{N}$, such that: $\forall x, y \in \Sigma^*$,
    \begin{align}
        \big \vert \textit{Pad}_{\mathcal{L}} (x,y) \big \vert \leq \textit{Poly}_{p} \big( \vert x \vert + \vert y \vert \big).
    \end{align}
    Then, for any fixed $x \in \mathcal{L}$, as $\textit{Pad}_{\mathcal{L}}$ by definition preserves membership of $\mathcal{L}$, $\mathcal{L}$ being sparse implies -- using Def.~\ref{def:sparseLanguages} -- that there exists a polynomial, $\textit{Poly}_{s}: \mathbb{N}_0 \rightarrow \mathbb{N}_0$, such that: $\forall n \in \mathbb{N}$,
    \begin{align}
        \label{eqn:sparsityBounds}
        \big \vert \big \{ \textit{Pad}_{\mathcal{L}} (x,y) \in \mathcal{L} \text{ } \vert \text{ } y \in \Sigma^{\leq n} \big \} \big \vert 
        \leq
        \sum_{j = 0}^{n} \bigg( \textit{Poly}_{s} \big( \textit{Poly}_{p} \big( \vert x \vert + j \big) \big) \bigg).
    \end{align}
     However, as $\textit{Pad}_{\mathcal{L}}: \Sigma^* \longrightarrow \Sigma^*$ is bijective (in its second input, as shown in Lemma~\ref{def:paddingBijective} in Appendix~\ref{App:paddIso}) and
     \begin{align}
         \big \vert \Sigma^{\leq n} \big \vert = \sum_{j = 0}^n \bigg( \big \vert \Sigma \big \vert^j \bigg) = \dfrac{ \vert \Sigma \vert^{n + 1} - 1 }{ \vert \Sigma \vert - 1 },
    \end{align}
    it follows that: $\forall n \in \mathbb{N}$,
     \begin{align}
        \label{Eqn:sizeOfSet}
         \big \vert \big \{ \textit{Pad}_{\mathcal{L}} (x,y) \in \Sigma^* \text{ } \vert \text{ } y \in \Sigma^{\leq n} \big \} \big \vert
         &=
         \dfrac{\vert \Sigma \vert^{n + 1} - 1}{\vert \Sigma \vert - 1}.
     \end{align}
     For convenience, I then define another polynomial, $\textit{Poly}_{sp}: \mathbb{N} \longrightarrow \mathbb{N}$, by: $\forall n \in \mathbb{N}$,
    \begin{align}
        \label{eqn:newPoly}
        \textit{Poly}_{sp}(n)
        &=
        \sum_{j = 0}^{n} \bigg( \textit{Poly}_{s} \big( \textit{Poly}_{p} \big( \vert x \vert + j \big) \big) \bigg).
    \end{align}
     Therefore, combining Eqn.~\ref{eqn:sparsityBounds}, Eqn.~\ref{Eqn:sizeOfSet}, and  Eqn.~\ref{eqn:newPoly} gives: $\forall n \in \mathbb{N}$,
     \begin{align}
         \label{eqn:contraDictStatement}
         \dfrac{\vert \Sigma \vert^{n + 1} - 1}{\vert \Sigma \vert - 1}
         \leq
         \textit{Poly}_{sp}(n).
     \end{align}
     As I assumed $\vert \Sigma \vert \geq 2$, there will always exist a sufficiently large $n \in \mathbb{N}$ to render Eqn.~\ref{eqn:contraDictStatement} false. Hence, the assertion in Eqn.~\ref{eqn:contraDictStatement} is a contradiction and $\mathcal{L} \subseteq \Sigma^*$ cannot exist as described. I.e. no language can be both paddable and sparse.
\end{proof}

\section{P-isomorphism Auxiliary Lemma}
\label{sec:PIsoDetailsAppendix}
The below Lemma~\ref{Lem:inverseIsPreserve} is used in the proof of Lemma~\ref{lem:PisoImpliesPhase}, to show that the third condition of phase transitions are met.
\begin{lemma}
    \label{Lem:inverseIsPreserve}
    For any preserving-P-isomorphism, $\xi: \Sigma^* \longrightarrow \Pi^*$, its inverse, $\xi^{-1}: \Pi^* \longrightarrow \Sigma^*$, is also a preserving-P-isomorphism.
\end{lemma}
\begin{proof}
    A preserving-P-isomorpism, as defined in Def.~\ref{def:preservingPIso}, is also a P-isomorphism.
    
    By definition (in Def.~\ref{def:PIsomorph}), if $\xi: \Sigma^* \longrightarrow \Pi^*$ is an P-isomorphism, so is its inverse, $\xi^{-1}: \Pi^* \longrightarrow \Sigma^*$. Then, using that for any $x \in \Sigma^*$ there exists a $y \in \Pi^*$ such that $y = \xi (x)$ (as $\xi$ is a bijection), alongside Eqn.~\ref{eqn:PreservingAspect} (in Def.~\ref{def:preservingPIso}): $\forall y \in \Pi^*$,
    \begin{align}
        \label{eqn:inversePreservingAspect}
        \xi^{-1}(y) \in \mathcal{L} \iff  x \in \mathcal{L} \iff \xi(x) \in \mathcal{H} \iff y \in \mathcal{H}.
    \end{align}
    So, as Eqn.~\ref{eqn:inversePreservingAspect} is the equivalent of Eqn.~\ref{eqn:PreservingAspect} but for $\xi^{-1}$, I conclude that $\xi^{-1}: \Pi^* \longrightarrow \Sigma^*$ is a preserving-P-isomorphism.
\end{proof}
\section{Padding Function Isomorphism Auxiliary Lemma}
\label{App:paddIso}
The aim of this appendix is to provide the statement and proof of the below Lemma~\ref{def:paddingBijective}, which is used in the proof of Lemma~\ref{lem:PisoImpliesPhase}, in Appendix~\ref{app:LemsForMainLemma} which in turn contributes to the proof of Theorem~\ref{lem:MainLemma}. 
\begin{lemma}
    \label{def:paddingBijective}
    For any alphabet, $\Sigma$, and paddable language, $\mathcal{L} \subset \Sigma^*$ the padding function of $\mathcal{L}$, $\textit{Pad}_{\mathcal{L}}: \Sigma^* \times \Sigma^* \longrightarrow \Sigma^*$, is bijective in its second argument. I.e. $\forall x \in \Sigma^*$, the function, $\textit{Pa}_{x}: \Sigma^* \longrightarrow \Sigma^*$, defined by: $\forall y \in \Sigma^*$,
    \begin{align}
        \textit{Pa}_{x}(y)
        &=
        \textit{Pad}_{\mathcal{L}} (x, y),
    \end{align}
    where $x \in \Sigma^*$ is fixed, is bijective.
\end{lemma}
\begin{proof}
    It is easy to see that, $\forall y \in \Sigma^*$, $\textit{Pa}_{x}(y)$ only returns a single value. It then only remains to show that $y$ can be uniquely recovered from $\textit{Pa}_{x}(y)$, which would then imply each input to $\textit{Pa}_{x}(y)$ is the only input that gives the output it gives and hence $\textit{Pa}_{x}(y)$ is a bijection.

    To show that $y$ can be uniquely recovered, I propose the following mapping as a left inverse of $\textit{Pa}_{x}(y)$: $\forall z \in \Sigma^*$,
    \begin{align}
        \textit{Pa}^{-1}_{x}(z)
        &=
        \textit{Dec}_{\mathcal{L}}(z),
    \end{align}
    where $\textit{Dec}_{\mathcal{L}}: \Sigma^* \longrightarrow \Sigma^*$ is the decoding function of $\mathcal{L}$, which must exist as it is paddable. This can be seen to correctly and uniquely recover $y$ as:
    \begin{align}
        \textit{Pa}^{-1}_{x}(\textit{Pa}_{x}(y))
        &=
        \textit{Dec}_{\mathcal{L}}( \textit{Pad}_{\mathcal{L}} (x, y))
        =
        y.
    \end{align}
\end{proof}


\end{document}